\documentclass{amsart}
\usepackage{graphicx}
\usepackage{amscd}
\usepackage{amsmath}
\usepackage{amsfonts}
\usepackage{amssymb}
\usepackage{bbm}
\usepackage{setspace}
\usepackage{enumerate}         % better lists
\usepackage{fixme}
\usepackage{color}
\usepackage{url}
\usepackage{amsthm}
\usepackage{bm}
\usepackage{xy}
\usepackage{enumitem}
\usepackage{mathrsfs}
\theoremstyle{plain}
\newtheorem{theorem}{Theorem}[section]

\newtheorem{corollary}[theorem]{Corollary}

\newtheorem{lemma}[theorem]{Lemma}
\newtheorem{notation}[theorem]{Notation}

\newtheorem{proposition}[theorem]{Proposition}

\newtheorem{definition}[theorem]{Definition}

\theoremstyle{remark}
\newtheorem{remark}[theorem]{Remark}

\numberwithin{equation}{section}

\newcommand{\ind}{1\!\kern-1pt \mathrm{I}}
\newcommand{\rsto}{]\!\kern-1.8pt ]}
\newcommand{\lsto}{[\!\kern-1.7pt [}

\numberwithin{equation}{section}

\renewcommand{\rho}{\varrho}

\DeclareMathOperator{\Diag}{Diag}

\DeclareMathOperator{\rk}{rank}

\begin{document}
\title[Polynomial processes in stochastic portfolio theory]{Polynomial processes in stochastic portfolio theory}

\begin{abstract}
We introduce polynomial processes in the sense of~\cite{CKT:12} in the context of stochastic portfolio theory to model simultaneously companies' market capitalizations and the corresponding market weights. These models substantially extend volatility stabilized market models considered by Robert Fernholz and Ioannis Karatzas in~\cite{FK:05}, in particular they allow for correlation between the individual stocks. At the same time they remain remarkably tractable which makes them applicable in practice, especially for estimation and calibration to high dimensional equity index data. In the diffusion case we characterize the joint polynomial property of the market capitalizations and the  corresponding weights, exploiting the fact that the transformation between absolute and relative quantities perfectly fits the structural properties of polynomial processes.
Explicit parameter conditions assuring the existence of a local martingale deflator \emph{and} relative arbitrages with respect to the market portfolio are given and the connection to non-attainment of the boundary of the unit simplex is discussed. We also consider extensions to models with jumps and the computation of optimal relative arbitrage strategies.
\end{abstract}


\author{Christa Cuchiero}
\address{Vienna University, Oskar-Morgenstern-Platz 1, A-1090 Vienna}
\email{christa.cuchiero@univie.ac.at}
\maketitle

\section{Introduction}

Since the seminal works by Robert Fernholz, e.g.~\cite{F:02, FS:82, F:99}, on \emph{stochastic portfolio theory} (SPT), many models to capture the joint behavior of companies' market capitalization, denoted by $S^i$ for $i=1, \ldots, d$, in large equity indices have been proposed. One of the main goals in this respect is to reproduce the empirically observed shape and dynamics of the capital distribution curves which have turned out to be remarkably stable over time (see \cite[Figure 5.1]{F:02}). Another major aspect is to analyze the relative outperformance of certain portfolios with respect to the market portfolio and in turn to develop models which allow for so-called relative arbitrages. The models used so far can be broadly categorized (see e.g.~the overview article \cite{V:15}) into 
\begin{itemize}
\item \emph{rank based} models with the \emph{Atlas} and the \emph{Hybrid Atlas} model \cite{BFK:05, IPBKF:11} as examples;
\item \emph{diverse} models~\cite{F:02, FKK:05};
\item \emph{sufficiently volatile} models with \emph{volatility stabilized} \cite{FKK:05} and \emph{generalized volatility stabilized models} \cite{P:14} as subclasses.
\end{itemize}
While rank based models were designed to match the empirically observed stability of the capital distribution curves over time,  
the definition of diverse and sufficiently volatile models hinges on certain descriptive features of equity markets which have turned out to lead to relative arbitrages. We here briefly comment on the main properties of these model classes.

As the whole field of SPT, rank based type models have attracted much attention from theoretical side and are a rich source of mathematically interesting questions, in particular in the field of interacting particle systems, which is an active area of research (see the references in \cite[Section 6]{V:15}). In terms of practical applicability there are however two shortcomings, namely the common assumption that capitalizations are uncorrelated (at least in the Atlas model) and limited tractability due to rank (in contrast to name) based modeling.

The second class of models, coined diverse, combine a condition that prevents the concentration of all the market capital into one single stock (so-called market \emph{diversity}) with \emph{strong non-degeneracy of the instantaneous covariance matrix} of the log-prices. R.~Fernholz's key insight was that under these two conditions it is possible to systematically outperform the market portfolio, even on arbitrary short time horizons (see \cite{FKK:05}). This combination is however not the only one under which relative arbitrages occur, in particular the strong non-degeneracy condition, which is usually hard to verify  in practice when estimating high dimensional covariance matrices, can be omitted.

Indeed, another condition,  which led to the introduction of sufficiently volatile models  is related to the volatility of the market weights,
more precisely to the \emph{cumulative excess growth rate} given by 
\begin{align}\label{eq1}
\Gamma^H(\cdot)=\frac{1}{2} \sum_{i=1}^d \int_0^{\cdot} \mu^i_t d\langle \log \mu^i \rangle_t\,,
\end{align}
where $\mu^i_t$ represents the market weight of the $i$-th stock at time $t \geq 0$, for $i = 1,\ldots,d$. Indeed, the assumption that the slope of $\Gamma^H(\cdot)$ is bounded away from zero is the defining property of this class of models. As proved by R.~Fernholz and I.~ Karatzas in \cite{FK:05} such models allow for relative arbitrages over sufficiently long time horizons with certain functionally generated portfolios,  but, as recently shown in \cite{FKR:16}, it is in general not sufficient to generate relative arbitrages on arbitrary short time horizons. 

Examples of sufficiently volatile models are so-called \emph{volatility stabilized market models} which have been introduced in~\cite{FK:05}. As stated in the overview article \cite{FK:09}, these models are remarkable for several reasons, in particular because
the total market capitalization process follows a specific Black Scholes model,
while the individual stocks reflect the empirical fact that log-prices of smaller stocks tend to have greater volatility than the log-prices of larger ones.
 
Despite this coherence with equity market features, there is one major drawback, namely again the lack of correlation between the individual stocks, for which reason their applicability to model realistic market situations is limited. To overcome this drawback we propose in the present paper a significant but still tractable extension of this model class, which we call \emph{polynomial market weight and asset price models}. 

Indeed, a further remarkable property of volatility stabilized models is the fact that the asset prices and the market weights follow each individually and also jointly a \emph{polynomial process.} Polynomial processes, introduced in \cite{CKT:12}, see also \cite{FL:14}, constitute a class of time-homogeneous Markovian-It\^o semimartingales which are inherently tractable in the sense that the calculation of (mixed) moments only requires the computation of matrix exponentials. 
The computational advantage associated with this property has been exploited in a large variety of problems. In particular, applications in mathematical finance include interest rates, credit risk, stochastic volatility models, life insurance liabilities  and variance swaps (see \cite{Delbaen/Shirakawa:2002, Ackerer/Filipovic:2016, ack_fil_pul_16, BZ:16, Filipovic/Gourier/Mancini:2015}). In the present context, it
can be utilized for implementing optimal arbitrages (see Section~\ref{sec:pmwapm}) as well as for model calibration to high dimensional equity indices, such as S\&P 500, using for instance method of moment or pathwise covariance estimation techniques exploiting the
analytical knowledge of moments as well as the specific functional form of the covariance structure. In a companion paper \cite{CGGPST:17} we actually investigate calibration of polynomial models to market data (MSCI world index with 300 stocks) and found very promising results. In particular, they are capable of matching the typical shape and fluctuations of capital distribution curves surprisingly well. A specificity of the setting of SPT, which differs from other applications in finance, is the choice of the market portfolio as 
num\'eraire, so that the market weights become the modeling quantity of primary interest. This thus necessitates tractable models whose state space is the \emph{unit simplex.} 
One is therefore naturally led to polynomial processes because the usually considered most tractable classes, namely L\'evy and affine processes, are deterministic on compact and connected state spaces (see \cite{KL:16}).

Summarizing, the class of polynomial market weight and asset price models provides a tool to high dimensional equity market modeling with 
\begin{itemize}
\item a high degree of tractability (indeed, they can be viewed as the most tractable class when it comes to modeling market weights directly);
\item the possibility of correlation between the stocks;
\item consistency with empirical market features, such as volatility structures that are inverse proportional to the size of the assets and ``correct'' fluctuations of the ranked market weights (capital distribution curves).
\end{itemize}

Therefore, one goal of the current article is to characterize these models, which are  essentially defined through the property that the joint process of weights and capitalizations $(\mu_t, S_t)_{t \geq 0}$ is polynomial. For diffusion models, this is the case if and only if the total captialization process, i.e. $\sum_{i=1}^d S^i$, is a polynomial diffusion on $\mathbb{R}_{++}$, independent of the polynomial diffusion for the market weights on the unit simplex (see Theorem \ref{th:main}). Furthermore, if the characteristics of $S$ are not allowed to depend on $\mu$, then 
$\sum_{i=1}^d S^i$ follows necessarily a Black Scholes model, slightly more general than in volatility stabilized models (see Corollary \ref{cor:main}). The crucial point is that the market weights process can have a much richer covariance structure than in the case of volatility stabilized models, where it corresponds to the \emph{Wright-Fisher diffusion} sometimes also called \emph{multivariate Jacobi process} (see e.g.~\cite{GJ:06}). In the present case the market weights process is a general polynomial diffusion on the unit simplex, as first characterized by Damir Filipovi\'c and Martin Larsson~\cite{FL:14} (see also \cite{CLS:16}). Despite this significantly higher flexibility, it shares many desirable properties of the subclass of volatility stabilized models. For instance, under certain parameter conditions as made precise in Theorem \ref{th:NANUPBR} the slope of the cumulative excess growth rate is bounded away from zero, whence polynomial market weight models can be subsumed under the class of sufficiently volatile ones. In particular, they allow for relative arbitrages over long time horizons, generated by the so-called entropic portfolio \cite[Example 11.1]{FK:09}.
Moreover, as in volatility stabilized models the existence of (strong) relative arbitrage on arbitrary short time horizons (even with long only portfolios) is also possible and can be deduced from~\cite{BF:08}.

Related to this is the precise characterization of the existence of relative arbitrages and a local martingale deflator (see e.g.~\cite{ST:14}, \cite{KKS:15}) which we provide in Section \ref{sec:relarbitrage}. As a byproduct we obtain precise conditions for boundary attainment of polynomial models on the unit simplex. In the case of non-attainment of certain boundary segments we then discuss computable approximate optimal arbitrage strategies based on polynomials. 

The remainder of the paper is structured as follows: Section~\ref{sec:setting} and Section~\ref{sec:poly}
are dedicated to introduce the general setting of stochastic portfolio theory and the precise notion of polynomial processes. In Section~\ref{sec:pmwapm} we then recall the definition of volatility stabilized market models, show their polynomial property and draw a full picture of polynomial market weight and asset price diffusion models, relying on a characterization of polynomial diffusions on $\Delta^d \times \mathbb{R}_+^m$ proved in Appendix \ref{app:pmwapm}. In Section \ref{sec:jumps}, we consider extensions with jumps, while Section \ref{sec:relarbitrage} is fully dedicated to analyze fundamental properties (existence of relative arbitrages, local martingale deflators, completeness) of polynomial market weight models.

\subsection{Notation}
For a stochastic process we shall usually write $X$ for $(X_t)_{t \geq 0}$ and 
use superscript indices for its components, with the exception of semimartingale characteristics where we use subscript indices. In this case the superscript index indicates the process to which they belong, e.g.~the characteristic $b^X$ denotes the drift of the process $X$. We denote by $\mathbb{N}$ the natural numbers, $\mathbb{N}_0:=\mathbb{N}\cup\{0\}$ the nonnegative integers, and $\mathbb{R}_+$ the nonnegative reals.
The symbols $\mathbb{R}^{n\times n}$, $\mathbb{S}^n$, and $\mathbb{S}^n_+$ denote the $n \times n$ real, real symmetric, and real symmetric positive semi-definite matrices, respectively. Furthermore, $e_i$ stands for the $i^{\text{th}}$ canonical unit vector 
and $\mathbf{1}$ is the vector whose entries are all equal to 1.

\section{Setting and notions of stochastic portfolio theory}\label{sec:setting} 

We start by recalling the general setting and several notions of stochastic portfolio theory (SPT). For a more detailed account we refer to \cite{F:02, FK:09, KR:16, FKR:16}.

\subsection{The market model and trading strategies}

Let $T >0$ denote some finite time horizon and let $(\Omega, \mathcal{F}, (\mathcal{F}_t)_{t \in [0,T]}, P)$ be a filtered probability space with a right continuous filtration. Thereon we consider, for $d \geq 2$,  an $\mathbb{R}_{+}^d$-valued semimartingale $S$ with $S_0 \in \mathbb{R}_{++}^d$, corresponding to the companies' (undiscounted) market capitalizations, i.e.~stock price multiplied by the number of outstanding shares. We denote by $\Sigma$ the total capitalization of the considered equity market, i.e.
\[
\Sigma= \sum_{i=1}^d S^i,
\]
and require that $\Sigma$ is a.s.~strictly positive, i.e. $P[ \Sigma_t >0, \forall  t \geq 0] =1$. In contrast to that the individual capitalizations are allowed to vanish.

In line with the common literature on SPT (in particular, \cite{KR:16}), we consider trading strategies investing only in these $d$ assets and do not introduce a bank account.

\begin{definition}\label{def:trading}
\begin{enumerate}
\item
We call a predictable, $\mathbb{R}^d$-valued, $S$-integrable process $\vartheta$ a \emph{trading strategy} with \emph{wealth process} $V^{v,\vartheta}$, where
\[
V_t^{v,\vartheta} =\sum_{i=1}^d \vartheta_t^i S_t^i, \quad 0 \leq t \leq T,  \quad \text{ with } V^{v,\vartheta}_0=v,
\]
is the time $t$-value of the investment according to the strategy $\vartheta$ and initial capital $v >0$. Each $\vartheta^i_t $ represents the number of shares held at time $t \geq 0$ in the $i^{\text{th}}$ asset.
\item A trading strategy $\vartheta$ is called \emph{self-financing} if $V^{v,\vartheta}$ satisfies
\[
V_t^{v,\vartheta}=v+\int_0^t \vartheta_s dS_s, \quad 0 \leq t \leq T,
\]
where the integral is understood in the sense of vector stochastic integration.
\item A trading strategy is called $v$-\emph{admissible} if 
$V^{v,\vartheta} \geq 0$ and $V^{v,\vartheta}_{-} \geq 0$.
\item A trading strategy is called \emph{long-only}, if it never sells any stock short, i.e. $\vartheta$ takes values in $\mathbb{R}^d_+$.
\end{enumerate}
We denote by $\mathcal{J}(S)$ the collection of all self-financing, $1$-admissible trading strategies.
\end{definition}

\begin{notation}
The canonical value for the initial capital will be $1$. For notational simplicity, we shall therefore write $V^{\vartheta}$ for $V^{1,\vartheta}$.
\end{notation}

If each component of the semimartingale $S$ is strictly positive, one can pass to a \emph{multiplicative} (as opposed to the above additive) modeling approach, which used to be
standard in SPT and can sometimes be convenient. In this setup one rather considers - instead of trading strategies corresponding to the number of shares - so-called portfolios whose value represents the proportion of current wealth invested in each of the assets (see 
Appendix \ref{app:multi} for further details). 

\subsection{Relative arbitrage with respect to the market}

A crucial quantity of interest is the performance of two admissible portfolios relative to each other. In this respect the notion of relative arbitrage defined subsequently plays an important role.

\begin{definition}[Relative arbitrage opportunity]\label{def:relArb}
Let $\varphi, \vartheta \in \mathcal{J}(S)$ be self-financing, $1$-admissible trading strategies. Then, the strategy $\vartheta$ is said to generate a \emph{relative arbitrage opportunity} with respect to $\varphi$ over the time horizon $[0,T]$ 
if 
\[
P\left[V^{\vartheta}_T \geq V_T^{\varphi}\right]=1 \textrm{ and } P\left[V_T^{\vartheta} > V_T^{\varphi}\right] >0.
\] 
Moreover, $\vartheta$ is called a strong relative arbitrage opportunity if 
\[
P\left[V_T^{\vartheta} > V_T^{\varphi}\right] =1.
\]
If the above definition applies to $\vartheta$ with values in $\mathbb{R}^d_+$, we speak of  \emph{(strong) long only relative arbitrage.}
\end{definition}

\begin{remark}\label{rem:v}
In the above definition we consider wealth processes with initial capital $1$ and $1$-admissible trading strategies. However, the existence of relative arbitrages clearly does not depend on that. Indeed, $\vartheta$ is a relative arbitrage opportunity with respect to $\varphi$ over the time horizon $[0,T]$ in the sense of Definition \ref{def:relArb} if and only if 
\begin{align}\label{eq:differentv}
P\left[V^{v,v \vartheta}_T \geq V_T^{ v, v \varphi}\right]=1 \textrm{ and } P\left[V_T^{v, v\vartheta} > V_T^{v,v\varphi}\right] >0.
\end{align}
Note that  $v \vartheta, v \varphi$ are self-financing, $v$-admissible trading strategies and the initial capital is $v$.
\end{remark}

We now draw our attention to \emph{relative arbitrage opportunities with respect to the market}. In this case the $1$-admissible reference trading strategy is given by $\varphi_t^i  = \frac{1}{\Sigma_0}$ for all $i \in \{1, \ldots,d\}$ and $t \in [0,T]$ and the above definition of relative arbitrage translates as follows:

\begin{definition}[Relative arbitrage opportunity ]\label{def:relArbmarket}
Let $\vartheta \in \mathcal{J}(S)$ be a self-financing, $1$-admissible trading strategy. Then $\vartheta$ is said to generate a \emph{relative arbitrage opportunity} with respect to the market over the time horizon $[0,T]$ 
if 
\[
P\left[V^{\vartheta}_T \geq \frac{\Sigma_T}{\Sigma_0}\right]=1 \textrm{ and } P\left[V_T^{\vartheta} > \frac{\Sigma_T}{\Sigma_0}\right] >0.
\] 
Moreover, $\vartheta$ is called a strong relative arbitrage opportunity with respect to the market if 
\[
P\left[V_T^{\vartheta} > \frac{\Sigma_T}{\Sigma_0}\right] =1.
\]
If the above definition applies to $\vartheta$ with values in $\mathbb{R}^d_+$, we speak of  \emph{(strong) long only relative arbitrage with respect to market.}
\end{definition}

\subsection{Market weights}
From Definition \ref{def:relArbmarket} it is apparent that relative arbitrages with respect to the market correspond to classical arbitrages when choosing the total capitalization $\Sigma$ as num\'eraire. Doing so, yields the process of \emph{market weights} denoted by $\mu=(\mu^1, \ldots, \mu^d)$ and defined by
\begin{align}\label{eq:weights}
\mu^i=\frac{S^i}{\Sigma}= \frac{S^i}{\sum_{i=1}^d S^i}, \quad i \in \{1, \ldots,d\}.
\end{align}
Since all $S^i$ are assumed to be nonnegative, $\mu$ takes values in the unit simplex $\Delta^d$ defined by
\[
\Delta^d=\left\{x \in [0,1]^d\, |\, \sum_{i=1}^d x_i=1\right\}.
\]
The interior of the unit simplex, i.e. $\{x \in (0,1)^d\, |\, \sum_{i=1}^d x_i=1\}$ is denoted by $\mathring{\Delta}^d$. Note that by assumption $\mu_0 \in \mathring{\Delta}^d$.

\begin{remark}\label{rem:relarb}
\begin{enumerate}
\item 
Let us remark that $\vartheta$ is a self-financing trading strategy for the capitalization process $S$ as of Definition \ref{def:trading} if and only if it is a self-financing trading strategy for the process $\mu$, where the latter is defined as in Definition \ref{def:trading} simply by replacing $S$ by $\mu$ (see \cite{GKR:95} and \cite[Proposition 2.3]{KR:16}). In particular, denoting by $Y^{q, \vartheta}$ the relative wealth process, i.e.
\begin{align}\label{eq:Y}
Y_t^{q, \vartheta}
=\frac{V_t^{q \Sigma_0, \vartheta}}{\Sigma_t},\quad 0 \leq t \leq T,  \quad \text{ with } Y_0^{q, \vartheta}=\frac{q \Sigma_0}{\Sigma_0}=q, 
\end{align}
we have
\[
Y_t^{q, \vartheta}= q+\int_0^t \vartheta_s d\mu_s, \quad 0 \leq t \leq T.
\]
As before we shall write $Y^{\vartheta}$ for $Y^{1,\vartheta}$.

\item Combining this with Definition \ref{def:relArbmarket} and using Remark \ref{rem:v} with $v=\Sigma_0$ in \eqref{eq:differentv} yields the following equivalent formulation of a relative arbitrage opportunity with respect to the market 
over the time horizon $[0,T]$, namely
\begin{align*}
P\left[Y^{\Sigma_0\vartheta}_T \geq 1\right]=1 \textrm{ and } P\left[Y_T^{\Sigma_0\vartheta} > 1\right] >0
\end{align*}
and 
\[
P\left[Y_T^{\Sigma_0\vartheta} > 1\right] =1
\]
for a strong relative arbitrage opportunity. The existence of relative arbitrages with respect to the market thus only depends on the market weights process $\mu$.
\end{enumerate}
\end{remark}

\section{Polynomial processes}\label{sec:poly}

This section is dedicated to give a very brief overview on polynomial processes which have been introduced in~\cite{CKT:12} and further analyzed in several papers, in particular in~\cite{FL:14}, where questions on existence and uniqueness of polynomial diffusions on different state spaces are treated. For the case of jump-diffusions on the specific state space $\Delta^d$, being of particular importance for modeling the market weight process $\mu$, we refer to \cite{CLS:16}.

We  here define  polynomial  processes  as  a  particular  class  of \emph{time-homogeneous Markovian It\^o-semimartingales} with state space $D \subseteq \mathbb{R}^n$ defined on a filtered probability space $(\Omega, (\mathcal{F}_t)_{t\in [0,T]}, \mathcal{F},P)$ with a right-continuous filtration and $T \in (0,\infty]$. 
By an It\^o-semimartingale $X$ we mean a semimartingale whose characteristics $(B^X, C^X, \nu^X)$ (with respect to a certain truncation function $\chi$) are absolutely continuous with respect to the Lebesgue measure, i.e., we have 
\[
B^X=\int_0^{\cdot} b^X_t dt, \quad C^X=\int_0^{\cdot} c^X_t dt, \quad \nu^X(dt, d\xi)=K^X_t(d\xi)dt, 
\]
and we call $(b^X,c^X,K^X)$ \emph{differential semimartingale characteristics}. The time-homogeneous Markov property is expressed by the fact  that  
$X$ is assumed to be Markovian relative to $(\mathcal{F}_t)$, that is,
\[
E[f(X_{t})|\mathcal{F}_s]=E[f(X_{t}) | \sigma(X_s)], \quad P\textrm{-a.s.}
\]
for all  $t \geq s$ and all Borel functions $f: D \rightarrow \mathbb{R}$ satisfying 
$E[|f(X_t)|]< \infty$ for all $t \in[0,T]$, where $\sigma(X_s)$ denotes the $\sigma$-algebra generated by $X_s$. In particular, the differential characteristics $(b^X,c^X,K^X)$ are functions $(b(X_t),c(X_t),K(X_t,d\xi))_{t \in [0,T]}$ of the current state of the process, where $b: D \to \mathbb{R}^n$,  $c: D \to \mathbb{S}^n_+$ are Borel functions and $K(x,d\xi)$ a Borel transition kernel from $D$ to $\mathbb{R}^n$.

Before giving the precise definition of a polynomial process, let us introduce some further notation. 
Let $\mathcal{P}_m$ denote the finite dimensional vector space of polynomials up to degree $m \geq 0$ on $D$, i.e.~the restrictions of polynomials on $\mathbb{R}^n $ to $D$, defined by
\begin{align*}
\mathcal{P}_m:=\left\{D \ni x \mapsto \sum_{|\mathbf{k}|=0}^m \alpha_{\mathbf{k}}x^{\mathbf{k}}, \,\Big |\, \alpha_{\mathbf{k}} \in \mathbb{R}\right\},
\end{align*}
where we use multi-index notation $\mathbf{k} = (k_1,\ldots, k_n) \in \mathbb{N}^n_0$, $|\mathbf{k}|=k_1+\cdots+k_n$ and $x^{\mathbf{k}}=x^{k_1}_1\cdots x^{k_n}_n$.
Moreover, we denote by $\mathcal{P}$ the vector space of all polynomials on $D$. Whenever  there is ambiguity to which state space $\mathcal{P}_m$ refers we denote it by $\mathcal{P}_m(D)$. 

The following definition of a polynomial process is not the original one given in~\cite{CKT:12}, but equivalent under the stated moment condition on the compensator of the jump measure as clarified by the subsequent theorem. For a slightly different definition of polynomial processes relying on solutions to the corresponding martingale problem and not a-priori assuming the existence of a Markov process, we refer to \cite{FL:14, CLS:16}.

\begin{definition}\label{def:poly} 
A \emph{polynomial process} $X$ is a time-homogeneous Markovian It\^o-semimartingale with state space $D$ and initial value $X_0=x \in D$ whose differential characteristics $(b^X_t,c^X_t,K^X_t(d\xi))_{t \in [0,T]}=(b(X_t),c(X_t),K(X_t,d\xi))_{t \in [0,T]}$, with respect to the ``truncation'' function $\chi(\xi)=\xi$, 
satisfy 
\[
E\left[\int_{\mathbb{R}^n} \| \xi \|^m K(X_t,d\xi)\big | X_0=x\right] < \infty \textrm{ for all } m \in \mathbb{N},\, x \in D, \, t \in [0,T]
\]
and
\begin{align*}
b_i(x) &\in \mathcal{P}_1 \quad i \in \{1, \ldots, n\},\\
c_{ij}(x)+\int_{\mathbb{R}^n} \xi_i \xi_j K(x,d\xi)&\in  \mathcal{P}_2 \quad i,j \in \{1, \ldots, n\},\\
\int_{\mathbb{R}^n} \xi^{\mathbf{k}} K(x,d\xi) &\in \mathcal{P}_{|\mathbf{k}|} \quad |\mathbf{k}| =3, \ldots.
\end{align*}
\end{definition}

The following theorem is a reformulation of the results in~\cite{CKT:12}.

\begin{theorem}
For a Markovian It\^o-semimartingale $X$ with state space $D$ and $E[\int_{\mathbb{R}^n} \| \xi \|^m K(X_t,d\xi)| X_0=x] < \infty$ for all $m \in \mathbb{N}$, $x \in D$, $t \in [0,T]$ the  following are equivalent:
\begin{enumerate}
\item $X$ is a polynomial process.
\item 
$x \mapsto E[f(X_t)| X_0=x] \in \mathcal{P}_k$ for all $k \in \mathbb{N}$, $f \in \mathcal{P}_k$, $x \in D$ and $t \in [0,T]$.
\end{enumerate}
\end{theorem}

\begin{proof}
The direction (i) $\Rightarrow$ (ii) is a direct consequence of~\cite[Theorem 2.15]{CKT:12}, while the (ii) $\Rightarrow$ (i)  follows from~\cite[Corollary 2.14]{CKT:12}. Note here that the original definition of a polynomial process corresponds to the statement in (ii). 
\end{proof}

\begin{remark}
The second assertion in the above theorem is crucial, since it means that the semigroup $(P_t)_{t \in [0,T]}$ associated with the Markov process maps $\mathcal{P}_k$ to $\mathcal{P}_k$. This in turn implies that expectations of polynomials of $X_t$ can be computed via matrix exponentials. More precisely
for every $k \in \mathbb{N}$, there exists a linear map $A$ on $ \mathcal{P}_k$, such that for all $t \in [0,T]$, the semigroup
$(P_t)$ restricted to $ \mathcal{P}_k$ can be written as
\[
P_t|_{ \mathcal{P}_k}=e^{tA}.
\]
This crucial property of polynomial processes allows for an easy and efficient computation of moments without knowing the probability distribution or the characteristic function.
\end{remark}

\begin{remark}
As indicated already in the introduction, another advantage of this model class arises in the context of model calibration in high dimensional situations, which we typically encounter when dealing with a large equity index.
Indeed, pathwise estimation techniques of the integrated or spot covariance can be applied to determine the parameters of the diffusion matrix $c$ whose entries are quadratic polynomials in the current state of the process (see Remark \ref{rem:marketweightcalibration} for further details).
\end{remark}

\begin{remark}\label{rem:martingaleprob}
As already mentioned above, an alternative approach to introduce polynomial processes is via the \emph{martingale problem} notion as for instance done in \cite{FL:14, CLS:16}. To this end, consider the following linear operator $\mathcal{G}: \mathcal{P} \to \mathcal{P}$ 
defined by
\begin{align*}
\mathcal{G}f(x)&=\sum_{i=1}^{n}D_if(x) b_i(x)+ \frac{1}{2} \sum_{i,j=1}^{n}D_{ij}f(x) c_{ij}(x)\\
 &\quad + \int_{\mathbb{R}^n} \left(f(x +\xi)- f(x) - \sum_{i=1}^{n}D_if(x) \xi_i\right)K(x,d\xi),
\end{align*}
with $(b,c,K)$ as of Definition \ref{def:poly}, which is usually called \emph{extended infinitesimal generator} (compare \cite[Definition 2.3]{CKT:12}). 

Given a probability distribution $\rho$ on $D$, 
a \emph{solution to the martingale problem} for $(\mathcal{G},\rho)$ is a c\`adl\`ag process $X$ with values in $D$ such that $P(X_0\in\cdot)=\rho$ and the process $M^f$ given by
\begin{equation*}\label{eq:Nf}
M_t^f:=f(X_t)-f(X_0)-\int_0^t\mathcal{G} f(X_s) ds
\end{equation*}
is a local martingale for every $f\in \mathcal{P}$ with respect to the filtration generated by $X$, i.e. $\mathcal{F}_t =\bigcap_{s >t} \mathcal{F}_s^X$ with $\mathcal{F}_s^X=\sigma(X_r\, | \, r \leq s)$.  We say that {\em the martingale problem for $\mathcal{G}$ is well-posed} if there exists a unique (in the sense of probability law) solution to the martingale problem for $(\mathcal{G},\rho)$ for any initial distribution $\rho$ on $D$. Since such a solution is a time-homogenous  Markovian It\^o-semimartingale (see \cite[Theorem 4.4.2]{EK:86} for the Markov property), we can therefore identify it with polynomial processes in the sense of Definition \ref{def:poly}.
\end{remark}

\section{Polynomial market weight and asset price models}\label{sec:pmwapm}

This section is dedicated to introduce models for the market weights and capitalizations  based on polynomial processes. We start by reviewing volatility stabilized market models as introduced in \cite{FK:05} and show that the joint process $(\mu, S)$ is polynomial in the sense of Definition \ref{def:poly} (see Proposition \ref{prop:jointproc}). Subsequently we characterize all polynomial diffusion models $(\mu, S)$ taking values in $\Delta^d \times \mathbb{R}^d_+$ where $\mu$ is Markovian in its own filtration and \eqref{eq:weights} holds true. The proofs of the results of this section can be found in Appendix \ref{app:pmwapm}.

\subsection{Volatility stabilized market models}\label{subsection:Vol}

The dynamics of the asset prices in volatility stabilized market models introduced in~\cite{FK:05} are defined through
\begin{align}\label{eq:vstm}
dS^i_t=S_t^i\left(\frac{1+\alpha}{2 \mu_t^i} dt + \frac{1}{\sqrt{\mu^i_t}} dW_t^i\right),\quad S_0^i=s^i, \quad i\in \{1, \ldots, d\},
\end{align}
where $\alpha \geq 0$ and $(W^1, \ldots,W^d)$ is a standard Brownian motion. We here and throughout the paper always consider weak solutions to such SDEs, or equivalently solutions to the associated martingale problems (compare Remark \ref{rem:martingaleprob}). Moreover, we suppose that the filtration is generated by $S$, i.e. 
\begin{align}\label{eq:filt}
\mathcal{F}_t =\bigcap_{s >t} \mathcal{F}_s^S \text{ with } \mathcal{F}_s^S=\sigma(S_r\, | \, r \leq s).
\end{align}
By a change of variable it is easily seen that these models are polynomial, which is stated in the proposition below. For its formulation recall the notation 
\[
\Sigma=\sum_{j=1}^d S^j.
\]

\begin{proposition}\label{prop:charS}
The volatility stabilized model $(S^1,\ldots, S^d)$ of~\eqref{eq:vstm} satisfies the following properties:
\begin{enumerate}
\item
$(S^1,\ldots, S^d)$ is a polynomial process on $\mathbb{R}_{++}^d$ whose differential characteristics $(b^S, c^S, K^S)$ are of the form
\begin{equation}
\begin{split}\label{eq:charS}
b^{S}_{i,t}&= \frac{1+\alpha}{2} \sum_{j=1}^d S_t^j, \\
c^{S}_{ii,t}&=S_t^i \left(\sum_{j=1}^d S_t^j \right),\quad c^{S}_{ij,t}=0,\\
K_t^S&=0.
\end{split}
\end{equation}
\item The dynamics of the total capitalization process are described by a Black-Scholes model of the form
\begin{align}\label{eq:BS}
d\Sigma_t=\Sigma_t\left(\frac{d(1+\alpha)}{2} dt+dZ_t\right),
\end{align}
where $Z$ denotes a Brownian motion.
\end{enumerate}
\end{proposition}

\begin{remark} \label{rem:small}
\begin{enumerate}
\item 
The name volatility stabilized model stems from the fact that the total capitalization process follows a specific Black Scholes model, while the individual assets have dynamics of the form~\eqref{eq:vstm}.
\item As already mentioned in the introduction the empirical feature that log-prices of smaller stocks tend to have greater volatility than the log-prices of larger ones is reflected in this model class. Indeed, from \eqref{eq:vstm} it is easily seen that we have $c^{\log S}_{ii}=\frac{1}{\mu^i}$, but however \emph{no correlation} as $c^{\log S}_{ij}=0$.
\end{enumerate}
\end{remark}

The following proposition asserts that also the market weights follow a polynomial process on the unit simplex, more precisely a so-called \emph{multivariate Jacobi process} (see, e.g.~\cite{GJ:06}). A similar statement has also been obtain in~\cite{G:09}
and ~\cite{P:11}. 
Let us introduce the filtration generated by $(\mu^1, \ldots, \mu^d)$ which we denote by $(\mathcal{G}_t)$, i.e.
\begin{align}\label{eq:Gfilt}
\mathcal{G}_t =\bigcap_{s >t} \mathcal{G}_s^{\mu} \textrm{ with } \mathcal{G}_s^{\mu}=\sigma(\mu_r\, | \, r \leq s). 
\end{align}
Note that $\mathcal{G}_t \subseteq \mathcal{F}_t$ since the information on $\Sigma$ is lost.

\begin{proposition} \label{prop:Jacobi}
In the volatility stabilized model of~\eqref{eq:vstm}, 
the dynamics of the market weights $(\mu^1,\ldots,\mu^d)$ can be described by a multivariate Jacobi process of the form
\begin{align}\label{eq:Jacobi}
d\mu^i_t=\left(\frac{1+\alpha}{2}-\frac{d(1+\alpha)}{2}\mu_t^i\right)dt+\sqrt{\mu_t^i}(1-\mu_t^i)dB_t^i-\sum_{i\neq j}\mu_t^i\sqrt{\mu_t^j}dB_t^j,
\end{align}
where $B$ is a $d$-dimensional standard Brownian motion.
In particular  $(\mu^1,\ldots,\mu^d)$ is a polynomial process with respect to  
$(\mathcal{G}_t)$ with state space $\mathring{\Delta}^d$ and differential characteristics 
$(b^{\mu}, c^{\mu}, K^{\mu})$ of the form
\begin{equation}\label{eq:charmu}
\begin{split}
b^{\mu}_{i,t}&= \frac{1+\alpha}{2} -\mu_t^i\frac{d(1+\alpha)}{2}, \\
c^{\mu}_{ii,t}&=\mu_t^i(1-\mu_t^i),\quad c^{\mu}_{ij,t}=-\mu_t^i\mu_t^j ,\\
K^{\mu}_t&=0.
\end{split}
\end{equation}
\end{proposition}

\begin{remark}\label{rem:posexcess}
\begin{enumerate}
\item As examples of sufficiently volatile models, volatility stabilized market models exhibit
a constant positive excess  growth rate $\gamma_{*}^{\mu}$ (compare with \eqref{eq1}), given by 
\[
\gamma_{*}^{\mu}=\frac{1}{2}\sum_{i=1}^d \mu^i d\langle \log \mu^i \rangle=\frac{1}{2}\sum_{i=1}^d \mu^i c^{\log(\mu)}_{ii}=
\frac{1}{2}\sum_{i=1}^d \frac{c^{\mu}_{ii}}{\mu^i}.
\]
From the above proposition it is easily seen that $\gamma_{*}^{\mu}=\frac{d-1}{2}$ and thus 
strong long-only relative arbitrages can be generated with the entropy function (see \cite[Example 11.1]{FK:09}). 
\item As recently shown in \cite{FKR:16}, a positive excess growth rate is in general not sufficient for the existence of relative arbitrages over \emph{arbitrary short} time horizons.
However, the class of volatility stabilized market models also allows for this type of arbitrage, as first proved in \cite{BF:08}. Indeed,  Assumption (2.8) in this paper, namely that the instantaneous variance of the minimal market weight satisfies  $c^{\mu}_{(d)(d)} \geq K \mu^{(d)}$ for some constant $K$,  is clearly satisfied as seen from \eqref{eq:charmu}. Here, the index $(d)$ denotes the minimum.
\end{enumerate}
\end{remark}

\begin{notation}\label{notationcov}
In the following proposition and subsequent theorems we write $c^{\mu,S}_{ij}$ for
the instantaneous covariance between  $\mu^i$ and $S^j$ and $c^{\Sigma, \mu}_{i}$ for the instantaneous covariance between $\mu^i$ and $\Sigma$ and similarly $c_i^{\Sigma, S}$.
\end{notation}

\begin{proposition}\label{prop:jointproc}
Consider the volatility stabilized model of~\eqref{eq:vstm}. Then, the joint process $(\mu,S)$ is a polynomial process with respect to $(\mathcal{F}_t)$ taking values in $\mathring{\Delta}^d \times \mathbb{R}^d_{++}$. In particular, the instantaneous covariance between  $\mu^i$ and $S^j$ as well as $\mu^j$ and $S^i$ for $i \neq j$ is given by 
\begin{align*}
c^{\mu,S}_{ij,t}= c^{\mu,S}_{ji,t}= -\mu_t^i S_t^j=-\mu_t^j S_t^i, \quad i \neq j
\end{align*}
and  between $S_t^i$ and $\mu_t^i$  by
\[
c^{\mu,S}_{ii,t}= S_t^i - \mu_t^i S_t^i.
\]
Moreover, the instantaneous covariance between $\Sigma$ and $\mu^i$ vanishes, i.e. $
c^{\mu,\Sigma}_{i,t}=0$ for all $i \in \{1, \ldots,d\}$. Hence $\mu$ and $\Sigma$ are independent.
\end{proposition}

\subsection{Definition and characterization of polynomial market weight and asset price models}\label{subsection:pmwapm}

Inspired by the above analyzed volatility stabilized models, we now aim to characterize all diffusion processes $(\mu,S)$ satisfying relation \eqref{eq:weights} and taking values in $\Delta^d \times \mathbb{R}^d_+$ such that $(\mu,S)$ is jointly polynomial and $\mu$ Markovian in its own filtration. The reason for assuming the latter is motivated by the fact that we are mainly interested in the relative performance with respect to the market and in modeling the capital distribution curves $\log k \mapsto \log \mu_{(k)}$. As outlined in Section \ref{sec:setting} modeling solely $\mu$ is sufficient for these purposes. Assuming $\mu$ to be Markovian is therefore reasonable, in particular in view of the high dimensional applications we are interested in.

We start by giving a formal definition of the class of polynomial market weight and asset price models:

\begin{definition}\label{def:polymodel}
Let $(\mu, S)$ be a stochastic process 
defined on a filtered probability space $(\Omega, (\mathcal{F}_t), \mathcal{F},P)$  taking values in  $\Delta^d \times \mathbb{R}_+^d $ such that $\Sigma >0$ and
$\mu^i = \frac{S^i}{\Sigma}$ for all  $i \in \{1,\ldots,d\}$. We call $(\mu, S)$ a \emph{polynomial market weight and asset price model} if 
\begin{enumerate}
\item the weights process $\mu$ is Markovian with respect to its natural filtration $(\mathcal{G}_t)$ (made right continuous) as defined in~\eqref{eq:Gfilt} and
\item the joint process $(\mu, S)$ is a polynomial process.
\end{enumerate}
\end{definition}

\begin{remark}
\begin{enumerate}
\item 
Note that the second requirement of the above definition is strong in the sense that for instance both components $\mu$ and $S$ can be polynomial processes but their covariance structure is of non-quadratic form so that the joint polynomial property is lost. We shall encounter such cases in Proposition~\ref{prop:main} below. Nevertheless the joint polynomial property is relevant when it comes to computing joint moments of $\mu$ and $S$.
\item Note that we cannot start the process at any point in $\Delta^d \times \mathbb{R}_+^d $ since \eqref{eq:weights} is required to hold true. In other words, the state space can be bijectively mapped to $\Delta^d \times \mathbb{R}_{++} $.
\end{enumerate}
\end{remark}

In order to characterize polynomial market weight and asset price models in terms of parameter restrictions let us introduce the following definition. Indeed, the parameter set introduced therein characterizes polynomial diffusion process on $\Delta^d$ (see Proposition \ref{prop:statespace} and  
\cite[Proposition 6.6]{FL:14}).

\begin{definition}\label{def:parameters}
We call a triplet $(\beta^{\mu}, B^{\mu}, \gamma^{\mu})$ with $\beta^{\mu} \in \mathbb{R}^d$, $B^{\mu} \in \mathbb{R}^{d\times d}$ and  $\gamma^{\mu} \in \mathbb{S}^d$ \emph{admissible simplex parameter set} if
\begin{itemize}
\item  $\gamma^{\mu}$ has nonnegative off-diagonal elements and all its diagonal elements are equal to $0$,
\item $(B^{\mu})^{\top} \mathbf{1}+(\beta^{\top}\mathbf{1})\mathbf{1}=0$ and $\beta^{\mu}_i+B^{\mu}_{ij}\geq 0$ for all $i,j\neq i \in \{1, \ldots,d\}$.
\end{itemize}
\end{definition}

We are now ready to state the announced characterization of polynomial market weight and asset price models with continuous trajectories.

\begin{theorem}\label{th:main}
The following assertions are equivalent:
\begin{enumerate}
\item The process $(\mu,S)$ is a polynomial market weight and asset price model with continuous trajectories.
\item The processes $\mu$ and $\Sigma$ are independent polynomial processes on $\Delta^d$ and $\mathbb{R}_{++}$ respectively.
\end{enumerate}
Each of the above conditions implies the following assertion: 
\begin{enumerate}
\item[(iii)]There exists an admissible simplex parameter set $(\beta^{\mu}, B^{\mu}, \gamma^{\mu})$ and parameters $\kappa, \phi \in \mathbb{R}_+$ satisfying $2 \kappa -\phi \geq 0$, $\lambda, \sigma \in \mathbb{R}$, such that the differential characteristics of $ (\mu, S, \Sigma)$ are given by 
\begin{align}
b^{\mu}_{i,t}&=\beta^{\mu}_i+\sum_{j=1}^d B^{\mu}_{ij}\mu_t^j, &\quad &b^{S}_{i,t}=\beta^{\mu}_i \Sigma_t+ \sum_{j=1}^d B^{\mu}_{ij}S_t^j +  \kappa \mu_t^i+ \lambda S_t^i,\notag\\ 
c^{\mu}_{ii,t}&=\sum_{i\neq j}\gamma^{\mu}_{ij} \mu_t^i \mu_t^j,&\quad & 
c^{\mu}_{ij,t}= -\gamma^{\mu}_{ij}\mu_t^i\mu_t^j, \quad i\neq j,\notag\\
c^{S}_{ii,t}&=\phi S_t^i\mu_t^i+\sigma^2 (S_t^i)^2+\sum_{j\neq i} \gamma^{\mu}_{ij} S_t^i S_t^j, &\quad &
c^{S}_{ij,t}=\phi S_t^i \mu^j+ \sigma^2 S_t^i S_t^j -\gamma^{\mu}_{ij} S_t^i S_t^j, \quad i \neq j,
\notag\\
c^{\mu, S}_{ii,t}&= \sum_{ j \neq i} \gamma^{\mu}_{ij} S_t^i\mu_t^j, &\quad 
&c^{\mu, S}_{ij,t}= c^{\mu, S}_{ji,t} =- \gamma^{\mu}_{ij} \mu_t^i S_t^j = - \gamma^{\mu}_{ij} \mu_t^j S_t^i, \quad i \neq j,
\notag\\
b^{\Sigma}_t&=\kappa+\lambda \Sigma_t, &\quad &
c^{\Sigma}_t= \phi \Sigma_t+ \sigma^2\Sigma^2_t,\notag\\ 
c^{\Sigma, S}_{i,t}&= \phi S_t^i+\sigma^2S_t^i \Sigma_t,&\quad &
c^{\Sigma, \mu}_{i,t}=0.\notag
\end{align}
\end{enumerate}

Conversely, for an admissible simplex parameter set $(\beta^{\mu}, B^{\mu}, \gamma^{\mu})$ and parameters $\kappa, \phi \in \mathbb{R}_+$, $\lambda, \sigma \in \mathbb{R}$, there exists a polynomial market weight and asset price model whose differential characteristics are of the above form.

\end{theorem}

\begin{remark}
Note that listing all characteristics of the process $(\mu, S, \Sigma)$ is of course redundant and they could be determined from the knowledge of the characteristics of $S$ or $(\mu, \Sigma)$. 
\end{remark}

\begin{remark}
\begin{enumerate}
\item 
In the case of volatility stabilized models the above parameters take the following values: $\beta^{\mu}=\frac{1+\alpha}{2} \mathbf{1}$, $B^{\mu}= -d\frac{1+\alpha}{2} I_d$, $\gamma^{\mu}_{ij} =1$ for all $i \neq j$, $\kappa= \phi =0$, $\lambda=d(\frac{1+\alpha}{2})$, $\sigma^2=1$. 
\item In comparison with Remark \ref{rem:small}, we see that the total capitalization process $\Sigma$ is not necessarily a Black-Scholes model, but can correspond to an affine process, more precisely a CIR process if $\sigma=0$.
\item Individual stocks can still reflect the fact that log-prices of smaller stocks tend to have greater volatility than the log-prices of larger ones, but allow additionally for correlation, in particular we have 
\[
c^{\log S}_{ii,t}=\frac{1}{\mu_t^i}\left(\frac{\sigma^2 S_t^i+\sum_{i\neq j}\gamma_{ij}^{\mu}  S_t^j+\phi \mu_t^i}{\Sigma_t}\right), \quad 
c^{\log S}_{ij,t}=-\gamma_{ij}^{\mu}+\sigma^2+ \frac{\phi}{\Sigma_t}.
\] 
This is the crucial advantage over volatility stabilized models.
\end{enumerate}
\end{remark}

\begin{remark}\label{rem:marketweightcalibration}
Let us here briefly comment on the practical applicability of these models in view of calibration, more precisely estimation of the instantaneous covariance of $\mu$. Note that the structure of $c^{\mu}$ allows to obtain $\gamma^{\mu}$ from any estimator of the integrated covariance of $\log{\mu}$ since 
\[
\frac{1}{T}\int_0^T c^{\log(\mu)}_{ij,t}dt=-\gamma_{ij}, \quad i\neq j.
\]
This simple relationship enables to estimate the parameters of the instantaneous covariance in any dimension and has been successfully implemented in \cite{CGGPST:17} to calibrate polynomial market weight and asset price models to market data of 300 stock constituting the MSCI World Index. In this specific case this means estimating 44 700 parameters of $\gamma^{\mu}$ which was under certain assumption on the correlation structure between the assets even possible on a scarse data set of a time series of around 300 time points.
As already mentioned in the introduction not only feasibility of this method could be proved but also the fact that polynomial models fit the data well and exhibit  promising empirical features in view of the shape and fluctuations of the capital distribution curves.
\end{remark}

The rather strong independence property between $\Sigma$ and $\mu$ which we obtained in the above theorem can be relaxed if we only require that each of the processes $\mu$ and $S$ is polynomial (but not necessarily jointly).  
Indeed in this case we have the following proposition.

\begin{proposition}\label{prop:main}
Consider an It\^o-diffusion process $S$ taking values in $\mathbb{R}_+^d $ such that $\Sigma= \sum_{i=1}^d S^i >0$. Let $(\mathcal{F}_t)$ be the filtration generated by $S$ as specified in \eqref{eq:filt}. Define  $\mu^i = \frac{S^i}{\Sigma}$ for all  $i \in \{1,\ldots,d\}$. 
Then the following assertions are equivalent:
\begin{enumerate}
\item  
Both processes $S$ and $\mu$ are polynomial (but not necessarily jointly).
\item The differential characteristics of $S$ are given by 
\begin{align*}
b^S_{i,t}&=  \beta^{\mu}_i \Sigma_t+ \sum_{k=1}^d B^{\mu}_{ik}S_t^k +   \lambda_i S_t^i,\\
c^{S}_{ii,t}&=(\zeta+2\lambda_i) (S_t^i)^2 +\sum_{k\neq i} \gamma^{\mu}_{ik} S_t^i S_t^k, &\quad &
c^{S}_{ij,t}=(\zeta +\lambda_i +\lambda_j -\gamma^{\mu}_{ij}) S_t^i S_t^j,
\end{align*} 
where $(\beta^{\mu}, B^{\mu}, \gamma^{\mu})$ is an admissible simplex parameter set and
and the parameters $\zeta \in \mathbb{R}$, $\lambda \in \mathbb{R}^d$ are such that 
$\zeta\mathbf{1}\mathbf{1}^{\top}+\Lambda -\gamma^{\mu} \in \mathbb{S}^d_+ $, where $\Lambda_{ij}=\lambda_i +\lambda_j$, in particular $\zeta+2\lambda_i \geq 0$ for all $i \in \ldots\{1, \ldots,d\}$.
\end{enumerate}

Moreover, each of the above conditions implies the following assertion. 
\begin{enumerate}
\item[(iii)] The differential characteristics of $(\mu, \Sigma)$ are given by
\begin{align*}
b^{\mu}_{i,t}&=\beta^{\mu}_i+\sum_{j=1}^d B^{\mu}_{ij}\mu_t^j, &\quad &
c^{\mu}_{ii,t}=\sum_{j\neq i}\gamma^{\mu}_{ij} \mu_t^i \mu_t^j,&\quad& 
c^{\mu}_{ij,t}= -\gamma^{\mu}_{ij}\mu_t^i\mu_t^j, \quad i\neq j,\\
b_t^{\Sigma}&= \sum_{i=1}^d \lambda_i S_t^i, &\quad &
c^{\Sigma}_t= \zeta \Sigma_t^2+ 2\Sigma_t \sum_{i=1}^d \lambda
_i S_t^i,   &\quad &
c^{\Sigma, \mu}_{i,t}=\lambda_i S_t^i -\mu_t^i\sum_{k=1}^d \lambda_k S_t^k,
\end{align*} 
and 
\begin{align*}
c^{\mu, S}_{ii,t}&= \sum_{j\neq i} \gamma^{\mu}_{ij} S_t^i\mu_t^j+\lambda_i \mu_t^i S_t^i-(\mu_t^i)^2\sum_{k=1}^d \lambda_k S_t^k,\\
c^{\mu, S}_{ij,t}&= - \gamma^{\mu}_{ij}  \mu_t^i S_t^j+ \lambda_i \mu_t^i S_t^j+\mu_t^i \mu_t^j \sum_{k=1}^d \lambda_k S_t^k, \quad i \neq j.
\end{align*}
\end{enumerate}

\end{proposition}

\begin{remark}
Note that for an admissible simplex parameter set $(\beta^{\mu}, B^{\mu}, \gamma^{\mu})$ and  parameters $\zeta \in \mathbb{R}$, $\lambda \in \mathbb{R}^d$  such that 
$\zeta\mathbf{1}\mathbf{1}^{\top}+\Lambda -\gamma^{\mu} \in \mathbb{S}^d_+ $, where $\Lambda_{ij}=\lambda_i +\lambda_j$,  the necessary conditions of Proposition \ref{prop:statespace} below for an $\mathbb{R}^d_+$-valued polynomial process are satisfied. Let us remark however, that only in specific parametric cases the well-posedness of the martingale problem for $S$ is known. This is for instance the case if $c^S_{ij}=0$, i.e. $\zeta+\lambda_i+\lambda_j=\gamma^{\mu}_{ij}$, which follows from \cite[Corollary 1.3]{BP:03}. Compare also with the assertions in \cite[Remark 6.5]{FL:14}.
\end{remark}

Finally combining Theorem \ref{th:main} and Proposition \ref{prop:main} yields the subsequent corollary stating that in polynomial market weight and asset price models with continuous trajectories, where $S$ is polynomial with respect to its own filtration, the total capitalization process $\Sigma$ is a Black Scholes model. So in this case we find a similar structure as in Proposition \ref{prop:charS} however with a much more general correlation structure between the assets.

\begin{corollary}\label{cor:main}
The following assertions are equivalent:
\begin{enumerate}
\item The process $(\mu,S)$ is a polynomial market weight and asset price model with continuous trajectories such that $S$ is polynomial with respect to its own filtration defined via \eqref{eq:filt}.
\item The processes $\mu$ and $\Sigma$ are independent polynomial processes on $\Delta^d$ and $\mathbb{R}_{++}$ respectively with the additional property that $\Sigma$ is a Black -Scholes model of the form
\[
d\Sigma_t=\lambda \Sigma_t dt + \sigma \Sigma_t dZ_t
\]
for some parameters $\lambda, \sigma \in \mathbb{R}$ and $Z$ a Brownian motion.
\end{enumerate}

\end{corollary}

\subsubsection{Extension with jumps}\label{sec:jumps}

In the following proposition we consider an extension with jumps which still gives rise to a polynomial market weight and asset price model.

\begin{proposition}\label{prop:jump}
Let $\mu$ and $\Sigma$ be independent polynomial processes (both components possibly with jumps) on $\Delta^d$ and  $\mathbb{R}_{++}$ respectively.
Assume that the respective jump measures satisfy 
\begin{equation}\label{eq:jumpmeas}
\begin{split}
\mu &\mapsto \int (\xi^{\mu})^\mathbf{k} K(\mu,d\xi^{\mu}) \in \mathcal{P}_{2}(\Delta^d) \text{ for } |\mathbf{k}|=2 \\
 \Sigma &\mapsto \int (\xi^{\Sigma})^2 K(\Sigma,d\xi^{\Sigma}) \in \mathcal{P}_2(\mathbb{R}_{++}).
\end{split}
\end{equation}
Define $S^i=\mu^i \Sigma$ for $i \in \{1,\ldots, d\}$.
Then $(\mu,S)$ is a polynomial market weight and asset price model.
\end{proposition}

\begin{remark}
\begin{enumerate}
\item
Note that the condition on the jump measure~\eqref{eq:jumpmeas} is anyhow satisfied for $|\mathbf{k}|\geq 3$ by the polynomial property of the processes $\mu$ and $\Sigma$.
\item For specifications of the jump structure of $\mu$ we refer to \cite[Section 6]{CLS:16}.
\item The assumptions on $\Sigma$ are satisfied if $\Sigma$ is for instance an affine CIR process with jumps.
\end{enumerate}
\end{remark}

\subsubsection{Polynomial processes on $\Delta^d \times \mathbb{R}^m_+$}

The proof of the above results relies on the following proposition which gives necessary conditions in terms of parameters for polynomial diffusion processes with state space $\Delta^d \times \mathbb{R}_+^m$. It is an extension of the results obtained in~\cite{FL:14}. 
For its formulation denote  $I=\{1, \ldots, d\}$ and $J=\{d+1, \ldots, d+m\}$. Moreover,  $x_I$ and $x_J$ stand for the vector $x$ consisting of the first $d$ and last $m$ elements respectively.

\begin{proposition}\label{prop:statespace}
Consider a polynomial diffusion process $X$ on $\Delta^d \times \mathbb{R}_+^m$, and denote  $I=\{1, \ldots, d\}$ and $J=\{d+1, \ldots, d+m\}$.
\begin{enumerate}
\item 
Then its diffusion matrix $c_t^X=c(X_t)$ is given by
\begin{align*}
c_{ii}(x)&=\sum_{i\neq j, 
j\in I}\gamma_{ij}x_ix_j, &\quad &(i \in I),\\
c_{ij}(x)&=-\gamma_{ij}x_ix_j, &\quad &(i,j \in I, i \neq j),\\
c_{ij}(x)&=0, &\quad& (i \in I, j \in J),\\
c_{jj}(x)&=\alpha_{jj}x_j^2+x_j(\phi_j+\theta_{(j)}^{\top} x_I+ \pi^{\top}_{(j)}x_J), &\quad& (j \in J),\\
c_{ij}(x)&=\alpha_{ij}x_ix_j, &\quad & (i,j \in J, i\neq j),
\end{align*}
where $\gamma_{ij}=\gamma_{ji} \in \mathbb{R}_+$, $\theta_{(j)} \in \mathbb{R}^d$, $\pi_{(j)} \in \mathbb{R}^m_+$ with $\pi_{(j)j}=0$, $\phi \in \mathbb{R}^m$ with $\phi_j \geq \max_{i\in I} \theta_{(j)i}^-$ and  $\alpha \in \mathbb{S}^m$ such that 
$\alpha+ \Diag (\Pi^{\top} x_J)\Diag(x_J)^{-1} \in \mathbb{S}^m_+$ for all $x_J \in \mathbb{R}^m_{++}$, where $\Pi \in \mathbb{R}^{m\times m}$ is the matrix with columns $\pi_{(j)}$.
\item The drift vector $b_t^X=b(X_t)$ satisfies
\begin{align*}
b(x)=\begin{pmatrix}
\beta_{I}+B_{II}x_I\\
\beta_J+ B_{JI}x_I+B_{JJ}x_J
\end{pmatrix}, 
\end{align*}
where $\beta \in \mathbb{R}^{d+m}$, $B \in \mathbb{R}^{d+m\times d+m}$ such that $B^{\top}_{II} \mathbf{1}+(\beta_I^{\top}\mathbf{1})\mathbf{1}=0$ and $\beta_i+B_{ij}\geq 0$ for all $i,j \in I$ with $j \neq i$, 
 $\beta_j \geq\max_{i \in I} B^-_{ji}$ for all $j \in J$ and $B_{JJ} \in \mathbb{R}^{m\times m}$ has nonnegative off-diagonal elements.
\end{enumerate}

\end{proposition}

\begin{remark}
Note that the drift part for the components in $\Delta^d$ could also be written as 
\[
\beta_{I}+B_{II}x_I=\widetilde{B}_{II} x_I,
\]
where $\widetilde{B}_{II,ij}=\beta_{I,i}+B_{II,ij}$ for all $j \in \{1,\ldots, d\}$ satisfies
$\widetilde{B}^{\top}_{II} \mathbf{1}=0$ and $B_{II,ij}\geq 0$ for all $i,j \in I, i \neq j$. 
In order to be consistent with the literature, we however keep the notation with a constant term.
\end{remark}

Under the above conditions on the parameters, one also gets existence of solutions to the martingale problem. However, well-posedness and hence the existence of polynomial processes in the sense of Definition \ref{def:poly} is not known in general. Under certain parameter restrictions this can be nevertheless achieved. In particular, the following holds and is relevant in our case for Theorem \ref{th:main} when $m=1$.

\begin{proposition}\label{prop:wellposed}
Consider the parameters given in Proposition \ref{prop:statespace}. Assume that $\alpha_{ij}=0$ for $i \neq j$. Then the martingale problem corresponding to the characteristics stated in Proposition \ref{prop:statespace} is well-posed.
\end{proposition}

\section{Relative arbitrage in polynomial models}\label{sec:relarbitrage}

This section is dedicated to characterize the existence of relative arbitrage opportunities
in polynomial diffusion market weight and asset price models under the so-called \emph{No unbounded profit with bounded risk condition} (NUPBR). Note that due to Remark \ref{rem:relarb}  relative arbitrages  only depend on the market weights process $\mu$. We thus only consider polynomial models on $\Delta^d$ and call them polynomial market weight models.
Recall that they are characterized in terms of an admissible simplex parameter set $(\beta^{\mu}, B^{\mu}, \gamma^{\mu})$. 

Recall that the (NUPBR) condition as introduced in \cite{DS:94} means that the set 
\[
\{ Y^{\vartheta}_T \,|\, \vartheta \in \mathcal{J}(\mu)\}
\]
is bounded in probability. Here, $\mathcal{J}(\mu)$  denotes the collection of all self-financing, 1-admissible trading strategies with respect to $\mu$ analogously as in Definition \ref{def:trading}. This (NUPBR) condition is the minimal requirement for economically reasonable models in continuous time and the usual assumption in stochastic portfolio theory. 

The following theorem provides a characterization of relative arbitrage under (NUPBR). Its proof together with the proofs of the subsequent propositions and lemmas 
are gathered in Appendix \ref{app:relarbitrage}. 

\begin{theorem}\label{th:NANUPBR}
Let $T \in (0, \infty)$ denote some finite time horizon and let $\mu$ be a polynomial diffusion process for the market weights on $\Delta^d$ with $\mu_0 \in \mathring{\Delta}^d$, being  described  by an admissible simplex parameter set $(\beta^{\mu}, B^{\mu}, \gamma^{\mu})$ with $\gamma^{\mu}_{ij} >0$ for all $i\neq j \in \{1, \ldots,d\}$. 
Then the following assertions are equivalent:

\begin{enumerate}
\item The model satisfies (NUPBR) 
and there exist strong relative arbitrage opportunities.
\item There exists some $i \in \{1, \ldots, d\}$ such that $b^{\mu}_i > 0$ for some elements in $\{\mu^i=0\}$ and for all such indices $i$, we have
\begin{align}\label{eq:nonbd}
2 \beta^{\mu}_i +\min_{i \neq j} (2 B^{\mu}_{ij} -\gamma^{\mu}_{ij}) \geq 0.
\end{align}
\end{enumerate}
\end{theorem}

\begin{remark} \label{rem:relarbitrage}
\begin{enumerate}
\item Note that \eqref{eq:nonbd} together with the strict positivity of $\gamma^{\mu}_{ij}$ implies in particular that $b^{\mu}_i > 0$ on $\{\mu^i=0\}$.
\item Similarly as volatility stabilized market models (see Remark \ref{rem:posexcess}), polynomial market weight models exhibit a positive excess  growth rate
\[
\gamma_{*}^{\mu}=\frac{1}{2}\sum_{i=1}^d \mu^i d\langle \log \mu^i \rangle=\frac{1}{2}\sum_{i=1}^d \mu^i c^{\log(\mu)}_{ii}=
\frac{1}{2}\sum_{i=1}^d \frac{c^{\mu}_{ii}}{\mu^i}\geq \min_{i \neq j} \gamma^{\mu}_{ij}\frac{d-1}{2},
\]
whenever $\gamma^{\mu}_{ij} >0$ for all $i\neq j$ and $\mu$ takes values  $\mathring{\Delta}^d$ where the latter is equivalent to \eqref{eq:nonbd} for all $i \in \{1, \ldots,d\}$.
Functionally generated relative arbitrage opportunities over sufficiently long time horizons can therefore be generated as in \cite[Example 11.1]{FK:09}. 
\item Under the condition that $\mu$ takes values  in  $\mathring{\Delta}^d$ and $\gamma^{\mu}_{ij} >0$ for all $i\neq j$, it can also be seen from \cite{BF:08} that strong long-only  relative arbitrages on arbitrary time horizons exist. Indeed,  Assumption (2.8) in this paper, namely that the instantaneous variance of the minimal market weight satisfies $c^{\mu}_{(d)(d)} \geq K \mu^{(d)}$ for some constant $K$ is satisfied due to
\begin{align*}
\frac{c^{\mu}_{(d)(d)}}{\mu^{(d)}}&= \sum_{j \neq (d) } \gamma^{\mu}_{(d)j} \mu^j \geq \min_{i \neq j} \gamma^{\mu}_{ij} (1-\mu^{(d)})\geq\frac{d-1}{d} \min_{i\neq j} \gamma^{\mu}_{ij},
\end{align*}
where we used $\mu^{(d)} \leq \frac{1}{d}$.
\end{enumerate}
\end{remark}

\begin{remark}\label{rem:martdensity}
It is well known that (NUPBR) is equivalent to the existence of a \emph{supermartingale deflator}, that is a nonnegative process $D$ with $D_0=1$ and $D_T >0$ such that $DY^{\vartheta}$ is a supermartingale for all $\vartheta \in \mathcal{J}(\mu)$ (see e.g.~\cite{KK:07}). More precisely, as shown in \cite{ST:14}, this can be strengthen to the existence of a \emph{strictly positive local martingale deflator} (see also \cite{KKS:15} and the references therein), that is a strictly positive local martingale $Z$  with $Z_0=1$ such that $Z\mu$ is a local martingale.
\end{remark}

The following proposition establishes completeness of polynomial market weight models and uniqueness of the strictly positive local martingale deflator $Z$, a property which is relevant for the construction of strong relative arbitrages subsequently. It is also needed in the proof of Theorem \ref{th:NANUPBR} to establish the existence of \emph{strong} relative arbitrages.

\begin{proposition}\label{prop:complete}
Let $\mu$ be a polynomial diffusion process for the market weights on $\Delta^d$ as of Theorem \ref{th:NANUPBR} satisfying one of the equivalent conditions (i) or (ii). Then the following assertions hold true: 
\begin{enumerate}
\item There exists a unique strictly positive local martingale deflator $Z$ as introduced in Remark \ref{rem:martdensity}.
\item The model is complete in the sense that every bounded $\mathcal{F}_T$-measurable claim $Y$ can be replicated via some strategy $\psi$. More precisely, 
\[
Y= E[YZ_T]+ \int _0^T \psi^{\top}_s d\mu_s, \quad P\text{-a.s.}
\]
holds true.
\end{enumerate}
\end{proposition}

\begin{remark}\label{rem:formdensity}
The strictly positive local martingale deflator $Z$ can be represented via $\mathcal{E}(-\int_0^{\cdot}\lambda^{\top}(\mu) d\mu^c)$, where $\mu^c$ denotes the martingale part of $\mu$ and  $\lambda$ is specified in Lemma \ref{lem:formlambda}. 
\end{remark}

The proof of Theorem \ref{th:NANUPBR} is based on the subsequent lemmas and propositions which are interesting in their own right. We start by the assertion that a continuous polynomial martingale with non degenerate diffusion matrix in the sense that $\gamma^{\mu}_{ij} >0$ for all $i \neq j$ reaches every boundary segment with positive probability on arbitrary time horizons.

\begin{lemma}\label{lem:martingaleboundary}
Let $\mu$ be a continuous polynomial martingale for the market weights on $\Delta^d$, being  described  by an admissible simplex parameter set $(\beta^{\mu}, B^{\mu}, \gamma^{\mu})$ with $\beta^{\mu}=0$ and $B^{\mu}=0$ and $\gamma^{\mu}_{ij} >0$ for all $i \neq j \in \{1, \ldots,d\}$. 
Then for any $T>0$, $k \in \{1,\ldots,d\}$, $\mu_0 \in \mathring{\Delta}^d$, we have $\mu^k_t = 0$ for some $t\leq T$ with positive probability.
\end{lemma}

The next proposition characterizes non-attainment of the boundary in terms of precise conditions on the admissible simplex parameter set.

\begin{proposition}\label{prop:nonattainement}
Let $\mu$ be a polynomial diffusion process for the market weights on $\Delta^d$ 
described by an admissible simplex parameter set $(\beta^{\mu}, B^{\mu}, \gamma^{\mu})$ with $\mu^i_0 >0$ $P$-a.s. 
Then the following assertions are equivalent:
\begin{enumerate}
\item For all $t >0$, $\mu^i_t >0$ $P$-a.s.
\item $2 \beta^{\mu}_i +\min_{i \neq j} (2 B^{\mu}_{ij} -\gamma^{\mu}_{ij}) \geq 0.$
\end{enumerate}
\end{proposition}

\begin{remark}
Observe that condition~\eqref{eq:nonbd} thus corresponds exactly to the non-attainment of $\{\mu_i=0\}$.
\end{remark}

After having characterized the existence of relative arbitrage opportunities let us now focus on their implementation. We shall consider so-called optimal arbitrages (see e.g.~\cite{FK:10}).

\begin{definition}\label{def:superhedge}
We denote by $U_T$ the superhedging price of $1$ at time $T >0$, that is,
\begin{align*}
U_T:=\inf\{q \geq 0 \, | \, \exists \, \vartheta \in \mathcal{J}(\mu)  \text{ with } Y_T^{q, \vartheta} \geq 1 \quad P\text{-a.s.}\}
\end{align*}
and 
we call \emph{$\frac{1}{U_T}$ optimal arbitrage}. 
\end{definition}

The relation to the strict local martingale deflator $Z$ of Proposition \ref{prop:complete} and the way how these optimal arbitrages can be implemented is described in the following remark. For technical reasons, needed in the assertions of Remark \ref{rem:optstrat} and the proof of Proposition \ref{prop:arbitrageexplicit}, we assume $(\Omega, (\mathcal{F}_t), \mathcal{F}, P)$ to be the canonical filtered probability space as for instance specified in \cite[Section 5]{R:13}.

\begin{remark}\label{rem:optstrat}
Consider a polynomial diffusion market weight models as of Theorem \ref{th:NANUPBR} and assume that (i) or equivalently (ii) holds. Let $J=\{j_1, \ldots, j_k\}$, $0 \leq k \leq d$,  denote the set of indices for which $b^{\mu}_{j_i}=0$ on $\{\mu^{j_i} =0\}$ and set $E:=\{ \mu \in \Delta^d\,|\, \mu^j >0 \text{ for all } j \notin J\}$. Then the optimal arbitrage can be achieved by investing $U_T$ and replicating the payoff
$1 \equiv 1_{E}(\mu_T)$ $P$-a.s., where the last identity is a consequence of Proposition \ref{prop:nonattainement}.  
More precisely, denote by $Z$ the strictly positive martingale deflator of Proposition \ref{prop:complete}. Then, by the superhedging duality (see e.g.~\cite[page 32]{KK:07} in the present context where only (NUPBR) holds) and completeness of the model as proved in Proposition \ref{prop:complete}, we have $U_T= E[Z_T]$ and the ``price at time $t$ of the optimal arbitrage'' is given by
\begin{align}\label{eq:price}
g(t, \mu_t)=\frac{E[Z_T| \mathcal{F}_t]}{Z_t}=\frac{E[1_{E}(\mu_T)Z_T| \mathcal{F}_t]}{Z_t}\stackrel{P\text{-a.s.}}{=} 
E_{Q}[1_{E}(\mu_T)| \mathcal{F}_t],
\end{align}
where $Q$ denotes the so called F\"ollmer measure (see~\cite{F:72, DS:95a, FKa:10, R:13}
and the references therein), for which $P \ll Q$ holds and under which $\mu$ is a martingale up to the time when $\mu$ leaves $E$. The last equality in \eqref{eq:price} follows from \cite[Theorem 5.1]{R:13} by noting that $\frac{1}{Z_t} >0$ $P$-a.s and $Q$-a.s.~on $E$. 
Assuming that $g$ is sufficiently regular, the replicating Delta hedging strategy $\vartheta$ is computed via 
\begin{align}\label{eq:thetaopt}
\vartheta^i_{t} = D_i g(t,\mu_t)
\end{align} 
(compare \cite[Theorem 4.1]{R:13}). Note that in the multiplicative setting (see Appendix \ref{app:multi}) this can be translated to portfolios weights given by
\[
\pi^i_{t}=\mu_t^i\left(\frac{D_i g(t, \mu_t)}{g(t,\mu_t)}+1 -\sum_{j=1}^d \mu_t^i \frac{D_i g(t, \mu_t)}{g(t,\mu_t)}\right).
\]
\end{remark}

In order to implement the above described strategy at least approximately, the polynomial property can be exploited. Indeed, the following proposition provides an approximation of the optimal arbitrage strategy in terms of polynomials, which can be easily implemented by approximating the function  $\mu \mapsto 1_{E}(\mu)$ via polynomials.

\begin{proposition} \label{prop:arbitrageexplicit}
Consider a polynomial diffusion market weight models as of Theorem \ref{th:NANUPBR} and assume that (i) or equivalently (ii) holds true.
Then for every $\varepsilon > 0$ there exists a time-dependent polynomial $\mu \mapsto p^{\varepsilon}(t,\mu)$ and a strategy defined via
\[
\vartheta^{i,\varepsilon}_{t} =D_i p^{\varepsilon}(t, \mu_t)
\]
such that $P[Y_T^{\vartheta^{\varepsilon}} > 1] \geq 1-\varepsilon$. Moreover, as $\varepsilon \to 0$, $Y_T^{\vartheta^{\varepsilon}}$ converges $P$-a.s. to the optimal arbitrage. 
\end{proposition}

\begin{remark}
Note that the strategy $\vartheta^{\varepsilon}_{t}$ yielding the ``approximate optimal arbitrage'' $Y_T^{\vartheta^{\varepsilon}}$
can be explicitly computed via matrix exponentials as it can be seen from \eqref{eq:vartheta}.
\end{remark}

\appendix

\section{Multiplicative modeling framework}\label{app:multi}
In addition to Section \ref{sec:setting}, let us here briefly review the \emph{multiplicative modeling framework} which used to be standard in stochastic portfolio theory
and can be applied if each component of the semimartingale $S$ is strictly positive. Indeed, then $S$ can be written in terms of the stochastic exponential of a $d$-dimensional semimartingale $R$ with $R_0=0$ and $\Delta R^i >-1$, i.e. $S^i=S^i_0\mathcal{E}(R^i)$, where $R$ is interpreted as the process of returns. Within this framework one can replace trading strategies by the notion of \emph{portfolios} defined as follows:

\begin{definition}
\begin{enumerate}
\item 
A \emph{portfolio} $\pi$ is a predictable process with values in
\[
 \left\{ x \in \mathbb{R}^d \, \Big{|}\, \sum_{i=1}^d x^i=1\right\}
\]
 such that 
$(\pi^1, \ldots, \pi^d)^{\top}$ is $R$-integrable. Each $\pi_t^i$ represents the proportion of current wealth invested at time $t$ in the $i^{\textrm{th}}$ asset for $i \in \{1, \ldots,d\}$.
\item A portfolio which satisfies $\pi^i \geq 0$ for all $i \in \{1, \ldots,d\}$, i.e., it takes values in the unit simplex $\Delta^d$
is called \emph{long-only}. 
\end{enumerate}
\end{definition}

Note that the market weights as defined in \eqref{eq:weights} are a particular long-only portfolio that invests in all assets in proportion to their relative weights.

By a slight abuse of notation we denote the wealth process achieved with initial wealth $v >0$ and by trading according to the portfolio $\pi$ by  $V^{v, \pi}$.
By converting the proportion of current wealth into numbers of shares, we can define a trading strategy $\vartheta$ by 
\[
\vartheta^i=\frac{V^{v,\pi}_{-} \pi^i }{S_{-}^i},\quad i=1, \ldots, d\,.
\]
The dynamics of $V^{v, \pi}$ can therefore be written as
\begin{align*}
\frac{dV_t^{v,\pi}}{V_{t-}^{v, \pi}}=\sum_{i=1}^n \pi_t^i \frac{dS_t^i}{S_{t-}^i}=\sum_{i=1}^n  \pi_t^i dR^i_t, \quad V_0^{v, \pi} =v >0.
\end{align*}
Similarly, the dynamics of the relative wealth process defined in \eqref{eq:Y} are in this setting given by
\begin{align*}
\frac{dY_t^{q,\pi}}{Y_{t-}^{q, \pi}}=\sum_{i=1}^d \pi_t^i\frac{d\mu_{t}^i}{\mu_{t-}^i}, \quad Y_0^{q, \pi} =q,
\end{align*}
in perfect analogy to the (original) wealth process where we have $\mu^i$ instead of $S^i$.

\section{Proofs of Section \ref{sec:pmwapm}} \label{app:pmwapm}

\subsection{Proofs of Subsection \ref{subsection:Vol}}

\begin{proof}[Proof of Proposition \ref{prop:charS}]
We start by proving (i). 
Recalling that  $ \mu^i=S^i/\Sigma$, we can rewrite~\eqref{eq:vstm} as
\begin{align}\label{eq:dynS}
dS^i_t= \frac{1+\alpha}{2} \Sigma_t dt+ \sqrt{S_t^i \Sigma_t} dW_t^i.
\end{align}
By~\cite[Theorem 1.2 and Corollary 1.3]{BP:03} (see also~\cite[Sections 4, 5]{FK:05}), this system of equations has a weak solution supported on $\mathbb{R}^d_{++}$, which is unique in the sense of probability law. Stated equivalently, this means that the associated martingale problem is well posed. Hence $(S^1, \ldots, S^d)$ is an It\^o-semimartingale which is Markovian, where the latter follows from~\cite[Theorem 4.4.2]{EK:86} due to the uniqueness of the solution to the martingale problem. Moreover, from~\eqref{eq:dynS} we see that that the drift and the diffusion matrix are linear and quadratic functions in the components of $S$. This thus yields the polynomial property and the form of the differential characteristics as stated in~\eqref{eq:charS}. 

Concerning (ii), note that the differential characteristics of $\Sigma$ are given by 
\[
(b^{\Sigma}, c^{\Sigma},  K^{\Sigma})=\left(d\frac{1+\alpha}{2}\Sigma, \Sigma^2, 0\right),
\] 
which implies that $\Sigma$ can be represented as in~\eqref{eq:BS}.
\end{proof}

\begin{proof}[Proof of Proposition \ref{prop:Jacobi}]
In order to compute the differential characteristics of $(\mu^1, \ldots, \mu^d)$, we apply~\cite[Proposition 2.5]{K:06}, which is simply a consequence of It\^o's formula, to the $C^2$-function $f: \mathbb{R}_{++}^{d+1} \to \mathring{\Delta}^d, f^i(S^1, \ldots, S^d, \Sigma)=S^i/\Sigma=\mu^i, i\in \{1, \ldots,d\}$. Denoting $\Sigma$ by $S^{d+1}$ and $\tilde{S}:=(S^1, \ldots, S^{d+1})$ as well as its differential characteristics accordingly, i.e.  
$b^{\Sigma}=b^{\tilde{S}}_{d+1}$ etc.,  we obtain from Proposition~\ref{prop:charS} and~\cite[Proposition 2.5]{K:06} the following form for the drift $b^{\mu}$
\begin{align*}
b^{\mu}_{i}&=D_i f^i(\tilde{S}) b^S_{i}+D_{d+1}f^i(\tilde{S}) b^{\tilde{S}}_{d+1}\\
&\quad +\frac{1}{2}\left(D_{ii} f^i(\tilde{S}) c^{S}_{ii}+2D_{i(d+1)}f^i(\tilde{S}) c^{\tilde{S}}_{i(d+1)}+D_{(d+1)(d+1)} f^i(\tilde{S})c^{\tilde{S}}_{(d+1)(d+1)}\right)\\
&=\frac{1}{S^{d+1}}\frac{1+\alpha}{2}S^{d+1}-\frac{S^i}{(S^{d+1})^2}\frac{d(1+\alpha)}{2}S^{d+1}\\
&\quad +\frac{1}{2}\underbrace{\left(-2\frac{1}{(S^{d+1})^2} S^{d+1}S^i+2 \frac{S^i}{(S^{d+1})^3} (S^{d+1})^2\right)}_{=0}\\
&=\frac{1+\alpha}{2}-\mu^i\frac{d (1+\alpha)}{2},
\end{align*}
and for the diffusion part $c^{\mu}$
\begin{align*}
c^{\mu}_{ii}&=(D_i f^i(\tilde{S}))^2c^{S}_{ii}+2D_i f^i(\tilde{S})D_{d+1} f^i(\tilde{S})c^{\tilde{S}}_{i(d+1)}+(D_{d+1} f^i(\tilde{S}))^2c^{\tilde{S}}_{(d+1)(d+1)}\\
&=\frac{1}{(S^{d+1})^2}S^{d+1}S^i-2\frac{1}{S^{d+1}}\frac{S^i}{(S^{d+1})^2}S^{d+1}S^i+\frac{(S^i)^2}{(S^{d+1})^4}(S^{d+1})^2\\
&=\mu^i-2(\mu^i)^2+(\mu^i)^2=\mu^i(1-\mu^i),\\
c^{\mu}_{ij}&=D_{i} f^i(\tilde{S})D_{d+1} f^j(\tilde{S})c^{\tilde{S}}_{i(d+1)}+ D_{d+1} f^i(\tilde{S})D_{j} f^j(\tilde{S})c^{\tilde{S}}_{j(d+1)}\\
&\quad +D_{d+1} f^i(\tilde{S})D_{d+1} f^j(\tilde{S})c^{\tilde{S}}_{(d+1)(d+1)}\\
&=-\frac{1}{(S^{d+1})}\frac{S^j}{(S^{d+1})^2}(S^i S^{d+1})- \frac{1}{(S^{d+1})}\frac{S^i}{(S^{d+1})^2}(S^j S^{d+1})\\
&\quad+\frac{S^i}{(S^{d+1})^2}\frac{S^j}{(S^{d+1})^2}(S^{d+1})^2\\
&=-\mu^i \mu^j, \quad i \neq j.
\end{align*}
Note here that $c^{\tilde{S}}_{i(d+1)}=S^iS^{d+1}$. This already yields the form of the differential characteristics in~\eqref{eq:charmu}. Noting that the differential characteristics of the Jacobi process given by~\eqref{eq:Jacobi} are the same and as the associate martingale problem admits a unique solution (see e.g.~\cite[Lemma 6.1]{CLS:16}), 
 we conclude that $(\mu^1, \ldots, \mu^d)$ defined via $\mu^i =S^i/\Sigma$ corresponds to this solution. Due to the fact that $(S^1, \ldots, S^d)$ takes values in $\mathbb{R}^d_{++}$, the state space of $(\mu^1,\ldots, \mu^d)$ is clearly $\mathring{\Delta}^d$. 
Moreover, since $S$ is Markovian with respect to $(\mathcal{F}_t)$, this is the case for $\mu$ as well. As the Markov property is preserved by passing to a coarser filtration to which $\mu$ is adapted to, in our case its natural filtration $(\mathcal{G}_t)$, 
 we can conclude the polynomial property of $(\mu^1, \ldots, \mu^d)$ with respect to $(\mathcal{G}_t)$.
\end{proof}

\begin{proof}[Proof of Proposition \ref{prop:jointproc}]
Note that in view of Proposition \ref{prop:charS} and \ref{prop:Jacobi}, we only have to prove that the instantaneous covariance between between $S$ and $\mu$ is a quadratic polynomial in the state variables. In order to compute this, we proceed as in the proof of Proposition \ref{prop:Jacobi}. Let $f: \mathbb{R}_{++}^{d+1} \to \mathring{\Delta}^d \times \mathbb{R}^{d+1}_{++}, f^i(S^1, \ldots, S^d, \Sigma)=S^i/\Sigma=\mu^i, f^{d+i}(S^1, \ldots, S^d, \Sigma)=S^i$ for $ i\in \{1, \ldots,d\}$ and $f^{2d+1}=\Sigma$ . Denoting $\Sigma$ by $S^{d+1}$ and its differential characteristics accordingly, and writing $\tilde{S}:=(S^1, \ldots, S^{d+1})$, 
 we have
\begin{align*}
c^{\mu,S}_{ij}&=D_{i} f^i(\tilde{S})D_{j} f^{d+j}(\tilde{S})c^{S}_{ij}+ D_{d+1} f^i(\tilde{S})D_{j} f^{d+j}(\tilde{S})c^{\tilde{S}}_{j(d+1)}\\
&=\frac{1}{S^{d+1}}c^{S}_{ij}-\frac{S^i}{(S^{d+1})^2}S^jS^{d+1}\\
&=\begin{cases}
S^i-\mu^i S^i & \text{if } i=j\\
-\mu^iS^j=-\mu^jS^i & \text{if } i\neq j.
\end{cases}
\end{align*}
Similarly, we have
\begin{align*}
c^{\mu,\Sigma}_{i}=c^{\mu,S^{d+1}}_{i}&= D_{i} f^i(\tilde{S})D_{d+1} f^{2d+1}(\tilde{S})c^{\tilde{S}}_{i(d+1)}\\
&\quad+ D_{d+1} f^i(\tilde{S})D_{d+1} f^{2d+1}(\tilde{S})c^{\tilde{S}}_{(d+1)(d+1)}\\
&=\frac{1}{S^{d+1}}S^i S^{d+1}-\frac{S^i}{(S^{d+1})^2}(S^{d+1})^2\\
&=0.
\end{align*}
This together with the form of the dynamics given in Proposition \ref{prop:charS} and \ref{prop:Jacobi}, implies independence of $\mu$ and $\Sigma$. 
\end{proof}

\subsection{Proofs of Section \ref{subsection:pmwapm}}

\begin{proof}[Proof of Theorem \ref{th:main}]
We prove the equivalence of (i) and (ii) by showing first that both conditions imply (iii).

Let us start by proving (i) $\Rightarrow$ (iii). By definition of a polynomial market weight model, $\mu$ is a polynomial process in its own filtration. Hence, by Proposition \ref{prop:statespace} taking $m=0$ (see also \cite[Proposition 6.6]{FL:14}), the form of the characteristic of $\mu$ follows immediately. Moreover, since $(\mu,S)$ is polynomial, the differential characteristics of $\Sigma$ necessarily satisfy 
\begin{align*}
b^{\Sigma}&=\kappa+\sum_{k=1}^d \lambda_k S^k+\sum_{k=1}^d \eta_k \mu^k, \\
c^{\Sigma}&=\alpha+\sum_{k=1}^d \phi_k S^k+ \sum_{k=1}^d \psi_k \mu^k+\sum_{k,l}\zeta_{kl} S^k S^l+\sum_{k,l} \theta_{kl} \mu^k S^l+ \sum_{k,l} \xi_{kl} \mu^k \mu^l,\\
c^{\Sigma, \mu}_i&=a^i+\sum_{k=1}^d A^i_k S^k+\sum_{k=1}^d B^i_k \mu^k+\sum_{k,l} C^i_{kl} S^k S^l+\sum_{k,l} D^i_{kl} \mu^k S^l+ \sum_{k,l} E^i_{kl} \mu^k \mu^l,
\end{align*}
for some parameters $\kappa, \alpha, a^i \in \mathbb{R}$, $ \lambda, \eta, \phi, \psi, A^i, B^i \in \mathbb{R}^d$ and $\zeta , \theta, \xi, C^i, D^i, E^i \in \mathbb{R}^{d \times d}$. We here assume without loss of generality that  $ \eta, \psi, B^i$, as well as the columns of $\theta$, $D^i$, $\xi$, $(\xi)^{\top}$,  $E^i$, $(E^i)^{\top}$ are not equal to $k\mathbf{1}$ for some constant $k \neq 0$, because then the linear and constant parts would be redundant. Let us now compute the differential characteristics of $S$. For the drift we have
\begin{align*}
b^S_i&= \Sigma b_i^{\mu}+ \mu^i b^{\Sigma} + c^{\Sigma, \mu}_i\\
&=\beta^{\mu}_i \Sigma + \sum_{k=1}^d B^{\mu}_{ik}S^k +  \kappa \mu^i+\mu^i(\sum_{k=1}^d \lambda_k S^k+\sum_{k=1}^d \eta_k \mu^k)+c^{\Sigma, \mu}_i.
\end{align*}
As the quadratic terms have to vanish, we obtain the following relationships, holding for all $i, k,l \in\{1, \ldots, d\}$ 
\begin{align*}
&C^i_{kl}=0 \\
&D^i_{ii}+\lambda_i=c,\quad D^i_{ik}+D^i_{ki}+\lambda_k=c, \quad k\neq i, 
\quad D^i_{kl}=0 \quad k \neq i \text{ and } l\neq i,\\
& E^i_{ii}+\eta_i=e,\quad E^i_{ik}+E^i_{ki}+\eta_k=e, \quad k\neq i,\quad E ^i_{kl}=0\quad k \neq i \text{ and } l\neq i,
\end{align*}
for some constants $c$ and $e$.
Therefore, 
\begin{align}
b^S_i&=  \beta^{\mu}_i \Sigma+ \sum_{k=1}^d B^{\mu}_{ik}S^k +  \kappa \mu^i+ a^i+\sum_{k=1}^d A^i_k S^k+\sum_{k=1}^d B^i_k \mu^k+cS^i+e \mu^i, \label{bS}\\
c^{\Sigma, \mu}_i&=a^i+\sum_{k=1}^d A^i_k S^k+\sum_{k=1}^d B^i_k \mu^k+\mu^i\left(\sum_{k=1}^d (c-\lambda_k) S^k+\sum_{k=1}^d (e-\eta_k) \mu^k\right) \label{csigmamu}.
\end{align}
For the instantaneous variance we have
\begin{align*}
c^{S}_{ii}&=(\mu^i)^2 c^{\Sigma}+2 \mu^i \Sigma c^{\Sigma, \mu}_i +\Sigma^2 c^{\mu}_{ii}\\
&=(\mu^i)^2\left(\alpha+\sum_{k=1}^d \phi_k S^k+ \sum_{k=1}^d \psi_k \mu^k+\sum_{k,l}\zeta_{kl} S^k S^l+\sum_{k,l} \theta_{kl} \mu^k S^l+ \sum_{k,l} \xi_{kl} \mu^k \mu^l\right)\\
&\quad +2 S^i\left(a^i+\sum_{k=1}^d A^i_k S^k+\sum_{k=1}^d B^i_k \mu^k+\mu^i\left(\sum_{k=1}^d (c-\lambda_k) S^k+\sum_{k=1}^d (e-\eta_k) \mu^k\right)\right)\\
&\quad +\sum_{k\neq i} \gamma^{\mu}_{ik} S^i S^k.
\end{align*}
In order to obtain a polynomial of degree $2$, the following conditions are necessary
\begin{equation}\label{eq:neccond}
\begin{split}
&\zeta_{kl}=\zeta + 2 \lambda_k, \\
&\theta_{kl}=\phi-\phi_l+ 2\eta_k, \\
&\xi_{kl}=\xi-\psi_k,  \\
\end{split}
\end{equation}
for all $l, k \in \{1, \ldots,d\}$. Here, $\zeta, \phi, \xi$ denote some constants. 
Inserting \eqref{eq:neccond} in the expression for $c^{\Sigma}$ yields
\begin{align*}
c^{\Sigma}&= \alpha+\xi+\phi \Sigma+ \zeta \Sigma^2+ 2\Sigma \sum_{k=1}^d \lambda
_k S^k  +2 \Sigma \sum_{k=1}^d  \eta_k \mu^k .
\end{align*}
Regarding $c^{\Sigma}$ as a function of $(\mu,S, \Sigma)$ it has to vanish whenever $(S, \Sigma)=0$, since $\Sigma$ is a strictly positive process.  
This implies that $\alpha+\xi=0$.
Since the instantaneous covariance matrix of $(\mu, \Sigma)$ has to be positive semidefinite, and both $c^{\mu}$ and $c^{\Sigma}$ do not contain a constant term, $a^i=0$ for all $i \in \{1,\ldots,d\}$. 
Hence, putting all together we obtain
\begin{align}\label{ciS}
c^{S}_{ii}=(\phi+2e)\mu^i S^i+(\zeta+2c)(S^i)^2+2 S^i(\sum_{k=1}^d A^i_k S^k+\sum_{k=1}^d B^i_k \mu^k)+\sum_{k\neq i} \gamma^{\mu}_{ik} S^i S^k.
\end{align}
For the instantaneous covariance we have
\begin{equation} \label{cijS}
\begin{split}
c^{S}_{ij}&=\mu^i \mu^j c^{\Sigma}+\mu^i \Sigma c^{\Sigma, \mu}_j+\mu^j \Sigma c^{\Sigma, \mu}_i+\Sigma^2 c^{\mu}_{ij}\\
&=\mu^i \mu^j (\phi \Sigma+ \zeta \Sigma^2+ 2\Sigma \sum_{k=1}^d \lambda
_k S^k  +2 \Sigma \sum_{k=1}^d  \eta_k \mu^k )\\
&\quad + S^i \left(\sum_{k=1}^d A^j_k S^k+\sum_{k=1}^d B^j_k \mu^k+\mu^j\left(\sum_{k=1}^d (c-\lambda_k) S^k+\sum_{k=1}^d (e-\eta_k) \mu^k\right)\right)\\
&\quad + S^j\left(\sum_{k=1}^d A^i_k S^k+\sum_{k=1}^d B^i_k \mu^k+\mu^i\left(\sum_{k=1}^d (c-\lambda_k) S^k+\sum_{k=1}^d (e-\eta_k) \mu^k\right)\right)\\
&\quad -\gamma^{\mu}_{ij} S^i S^j\\
&=\phi S^i \mu^j+ \zeta S^i S^j+ S^i \left(\sum_{k=1}^d A^j_k S^k+\sum_{k=1}^d B^j_k \mu^k+c S^j +e \mu^j\right)\\
&\quad + S^j\left(\sum_{k=1}^d A^i_k S^k+\sum_{k=1}^d B^i_k \mu^k+cS^i +e \mu^i\right) -\gamma^{\mu}_{ij} S^i S^j.
\end{split}
\end{equation}
Finally, we compute the instantaneous covariance  between $\mu$ and $S$
\begin{equation}
\begin{split} \label{cmuSi}
c^{\mu,S}_{ii}&=\mu^i c^{\Sigma,\mu}_i+\Sigma c^{\mu}_{ii}\\
&=\mu^i \left(\sum_{k=1}^d A^i_k S^k+\sum_{k=1}^d B^i_k \mu^k+\mu^i\left(\sum_{k=1}^d (c-\lambda_k) S^k+\sum_{k=1}^d (e-\eta_k) \mu^k\right)\right)\\
&\quad +S^i\sum_{j \neq i} \gamma^{\mu}_{ij} \mu^j.\\
c^{\mu,S}_{ij}&=\mu^j c^{\Sigma,\mu}_i+\Sigma c^{\mu}_{ij}\\
&=\mu^j \left(\sum_{k=1}^d A^i_k S^k+\sum_{k=1}^d B^i_k \mu^k+\mu^i\left(\sum_{k=1}^d (c-\lambda_k) S^k+\sum_{k=1}^d (e-\eta_k) \mu^k\right)\right)\\
&\quad - \gamma^{\mu}_{ij} S^i\mu^j .
\end{split}
\end{equation}
In order to obtain a polynomial of degree 2, we necessarily have for all $k$, $\lambda_k \equiv \lambda$ where $\lambda$ denotes now (by a slight abuse of notation) some constant. Similarly $\eta = k \mathbf{1}$ for some $k$ but the assumption at the beginning that   $\eta \neq k \mathbf{1}$  for $k\neq 0$ implies that $\eta=0$. This together with $\alpha+\xi=0$ yields
\begin{equation} \label{eq:charSigma}
\begin{split}
c^{\Sigma}&=\phi \Sigma+ (\zeta +2\lambda)\Sigma^2 \\
b^{\Sigma}&=\kappa + \lambda \Sigma.
\end{split}
\end{equation}
In order obey the condition that $\Sigma = \sum_{i=1}^d S^i$, we need to verify $b^{\Sigma}=\sum_{i=1}^d b^S_i$ and $c^{\Sigma}= \sum_{i,j} c^S_{ij}$.
This leads to the following conditions
\begin{align*}
A^i_i&=\lambda-c, \quad A^i_k=0, \quad i \neq k,\\
B^i_i&=-e, \quad B^i_k=0, \quad i \neq k,
\end{align*}
and implies due to \eqref{csigmamu}
\begin{align}\label{eq:cSigmamuneu}
c^{\Sigma, \mu}_i&=0.
\end{align}
Therefore, we see from the form of the characteristics for $(\mu, \Sigma)$, that it is a polynomial process on $ \Delta^d \times \mathbb{R}_+$. In particular, \eqref{eq:cSigmamuneu} is in line with Proposition \ref{prop:statespace}.  
Moreover, from Proposition \ref{prop:statespace} it follows that $\kappa, \phi$ and $\zeta + 2\lambda$ have to be non-negative. Since $\Sigma$ is strictly positive, Lemma \ref{lem:strictpos} further implies the condition $2\kappa -\phi \geq 0$.

In order to obtain the final form of the characteristics, we now insert all these restrictions in the equations \eqref{bS}, \eqref{ciS}, \eqref{cijS} and \eqref{cmuSi}. The remaining characteristics involving $\Sigma$ are given by \eqref{eq:charSigma} and \eqref{eq:cSigmamuneu}, and $c^{\Sigma, S}$ results from a straightforward computation. Observe that the only expressions involving $\zeta$ are $\zeta +2 \lambda $ so that we can replace $\zeta+2 \lambda\geq 0$ by some parameter $\sigma^2\geq0$. This then yields the form of the characteristics in assertion (iii) with the corresponding admissible parameters. 

In order to prove (i) $\Rightarrow$ (ii), it only remains to prove the independence of $\mu$ and $\Sigma$. 

In this respect note that the martingale problem in the sense of Remark \ref{rem:martingaleprob} associated to the characteristic of $\mu$ and $\Sigma$ (see e.g.~\cite[Lemma 6.1]{CLS:16}, \cite[Corollary 1.3]{BP:03} or Proposition \ref{prop:wellposed}) is well-posed. Therefore the solution corresponds to a weak solution of 
\begin{align*}
d\mu_t&=b^{\mu}_t dt + \sqrt{c^{\mu}_t} dB_t,\\
d\Sigma_t &= (\kappa +\lambda \Sigma_t ) dt+ \sqrt{\phi\Sigma_t + \sigma^2 \Sigma_t} dZ_t,
\end{align*}
where $B$ is a $d$-dimensional standard Brownian motion and  $Z$ a one-dimensional one,  independent of $B$. As $\mu$ and $\Sigma$ are fully decoupled, independence follows.

Let us now turn to the converse direction (ii) $\Rightarrow$ (i) by first proving that (ii) implies (iii). Observe that the existence of an admissible simplex parameter set and parameters $\kappa, \phi \in \mathbb{R}_+$ with $2\kappa +\phi \geq 0$, $\lambda, \sigma \in \mathbb{R}$ determining the form of the characteristics of $\mu$ and $\Sigma$ follows immediately from Proposition \ref{prop:statespace} and Lemma \ref{lem:strictpos}. Note in this respect that the independence of $\mu$ and $\Sigma$ implies that both $b^{\Sigma}$ and $c^{\Sigma}$ do not depend on $\mu$. The remaining characteristics can then be easily computed via It\^o's product rule using the fact that $S^i=\mu^i \Sigma$. This then yields the expressions stated in (iii), which implies that the joint process $(\mu,S)$ is polynomial.  
This together with the fact that $\mu$ is clearly a polynomial process in its own filtration allows to conclude that $(\mu,S)$ is a polynomial market weight and asset price model with continuous trajectories, whence (i).

The last statement of the theorem follows again from the fact that the martingale problem 
corresponding to the characteristics of $\mu$ and  $\Sigma$ as given in (iii) is well-posed as long as the involved parameters satisfy the stated admissibility conditions. The existence of a polynomial market weight and asset price model can then be deduced as in the proof of direction (ii) $\Rightarrow$ (i).
\end{proof}

The following lemma is needed in the above proof to characterize strict positivity of one dimensional polynomial processes.

\begin{lemma}\label{lem:strictpos}
Let $\Sigma$ be a polynomial diffusion process on $\mathbb{R}_{+}$ with $\Sigma_0> 0$, $b_t^{\Sigma}=\kappa +\lambda \Sigma_t$ and  $c_t^{\Sigma}=\phi \Sigma_t+\sigma^2 \Sigma_t^2$ where $\kappa, \phi \geq 0$ and $\lambda, \sigma \in \mathbb{R}$. Then the following assertions are equivalent:
\begin{enumerate}
\item For all $t >0$, $\Sigma_t >0$ $P$-a.s. 
\item $2\kappa -\phi \geq 0$.
\end{enumerate}
\end{lemma}

\begin{proof}
To prove (ii) $\Rightarrow$(i) we apply McKean's argument (see e.g.~\cite[Proposition 4.3]{MPS:11}) to the $\log(\Sigma)$. 
Then by It\^o's formula, we have for $t < \tau :=\inf\{ s \geq 0 \, |\, \Sigma_s =0\}$
\begin{align*}
\log(\Sigma_t)=\log(\Sigma_0) + \int_0^t \left(\frac{2\kappa -\phi}{2 \Sigma_s} +\lambda- \frac{1}{2}\sigma^2 \right)ds +\int_0^t \frac{\sqrt{\phi \Sigma_s+ \sigma^2 \Sigma_s^2}}{\Sigma_s} dW_s,
\end{align*}
where $W$ denotes some one-dimensional Brownian motion. 
Since the first  term in the drift is nonnegative and the second is constant, we deduce that for every $\mathbf{T} >0$ 
\[
\inf_{t \in [0, \tau \wedge \mathbf{T})} \int_0^t \left(\frac{2\kappa -\phi}{2 \Sigma_s} +\lambda- \frac{1}{2}\sigma^2 \right)ds  > -\infty, \quad P\text{-a.s.}
\]
Mc Kean's argument as of e.g.~\cite[Proposition 4.3]{MPS:11} therefore yields that $\tau =\infty$ and thus implies the assertion.

For the converse direction we apply \cite[Theorem 5.7 (iii)]{FL:14} and assume for a contradiction that (ii) does not hold, i.e.~$2 \kappa -\phi <0$. In the terminology of \cite[Theorem 5.7 (iii)]{FL:14}, we have $p(x)=x$, $\bar{x}=0$, $h(x)=\phi+\sigma^2 x$.
Since 
\begin{align*}
\mathcal{G}p(\bar{x}) =\mathcal{G}p(0) =\kappa\geq 0 \text{ and } 2 \mathcal{G} p(\bar{x})- h(\bar{x}) p'(\bar{x}) = 2\kappa-\phi<0, 
\end{align*}
\cite[Theorem 5.7 (iii)]{FL:14} thus implies that for any time horizon $T$ we can find some $\Sigma_0>0$ close enough to $0$ such that $0$ is hit with positive probability.
This contradicts (i) and proves the assertion.

\end{proof}

\begin{proof} [Proof of Proposition \ref{prop:main}] 
Let us start by proving (i) $\Rightarrow$ (ii). As $\mu$ is supposed to be a polynomial process,  the form of the characteristic of $\mu$ as stated in (iii) follows similarly as in Theorem \ref{th:main} from Proposition \ref{prop:statespace}. Moreover, since $S$ is supposed to be polynomial too, the differential characteristics of $\Sigma$ necessarily satisfy  
\begin{align*}
b^{\Sigma}&=\kappa+\sum_{k=1}^d \lambda_k S^k, \\
c^{\Sigma}&=\alpha+\sum_{k=1}^d \phi_k S^k+\sum_{k,l}\zeta_{kl} S^k S^l,\\
\end{align*}
for some parameters $\kappa, \alpha \in \mathbb{R}$, $ \lambda,\phi \in \mathbb{R}^d$ and $\zeta \in \mathbb{R}^{d \times d}$. Using these expressions and the characteristics of $\mu$ we now compute the differential characteristics of $S^i=\mu^i \Sigma$ for $i \in \{1, \ldots,d\}$. For the drift we have
\begin{align*}
b^S_i&= \Sigma b_i^{\mu}+ \mu^i b^{\Sigma} + c^{\Sigma, \mu}_i\\
&=\beta^{\mu}_i \Sigma + \sum_{k=1}^d B^{\mu}_{ik}S^k +  \kappa \mu^i+\mu^i(\sum_{k=1}^d \lambda_k S^k)+c^{\Sigma, \mu}_i.
\end{align*}
In order to obtain an affine function in the components of $S$, $c^{\Sigma, \mu}$ is necessarily of the following form
\begin{align}\label{eq:covmuSig}
c^{\Sigma, \mu}_i=a^i -\kappa \mu^i+ \sum_{k=1}^d A^i_k S^k +\mu^i\sum_{k=1}^d (c-\lambda_k) S^k.
\end{align}
Therefore, 
\begin{align}
b^S_i&=  \beta^{\mu}_i \Sigma+ \sum_{k=1}^d B^{\mu}_{ik}S^k +  a^i+\sum_{k=1}^d A^i_k S^k+cS^i\label{bS1}.
\end{align}
For the instantaneous variance we have
\begin{align*}
c^{S}_{ii}&=(\mu^i)^2 c^{\Sigma}+2 \mu^i \Sigma c^{\Sigma, \mu}_i +\Sigma^2 c^{\mu}_{ii}\\
&=(\mu^i)^2(\alpha+\sum_{k=1}^d \phi_k S^k+\sum_{k,l}\zeta_{kl} S^k S^l)\\
&\quad +2 S^i\left(a^i-\kappa \mu^i+\sum_{k=1}^d A^i_k S^k+\mu^i\left(\sum_{k=1}^d (c-\lambda_k) S^k\right)\right)\\
&\quad +\sum_{k\neq i} \gamma^{\mu}_{ik} S^i S^k.
\end{align*}
To obtain a polynomial of degree $2$, that does not depend on $\mu$ the following conditions are necessary
\begin{align*}
&\alpha=0\\
&\zeta_{kl}=\zeta+2\lambda_k,  \\
&\phi_k=2\kappa, 
\end{align*}
for all $l, k \in \{1, \ldots,d\}$. Here, by a slight abuse of notation $\zeta$ denotes some constant. Inserting these restrictions into $c^{\Sigma}$ yields
%\begin{align*}
$c^{\Sigma}= 2\kappa \Sigma+ \zeta \Sigma^2+ 2\Sigma \sum_{k=1}^d \lambda
_k S^k$. 
%\end{align*}
 Since the instantaneous covariance matrix of $(\mu, \Sigma)$ has to be positive semidefinite, and both $c^{\mu}$ and $c^{\Sigma}$ do not contain a constant term, $a^i = 0$ for all $i \in {1,...,d}$ in \eqref{eq:covmuSig}.
We then have
\begin{align}\label{ciiS1}
c^{S}_{ii}&= (\zeta+2c) (S^i)^2+2 S^i\left(\sum_{k=1}^d A^i_k S^k\right) +\sum_{k\neq i} \gamma^{\mu}_{ik} S^i S^k
\end{align}
and for the instantaneous covariance $c^{S}_{ij}$
\begin{equation} \label{cijS1}
\begin{split}
c^{S}_{ij}&=\mu^i \mu^j c^{\Sigma}+\mu^i \Sigma c^{\Sigma, \mu}_j+\mu^j \Sigma c^{\Sigma, \mu}_i+\Sigma^2 c^{\mu}_{ij}\\
&=\mu^i \mu^j ( 2\kappa \Sigma+ \zeta \Sigma^2+ 2\Sigma \sum_{k=1}^d \lambda
_k S^k)\\
&\quad + S^i \left(-\kappa \mu^j+ \sum_{k=1}^d A^j_k S^k +\mu^j\sum_{k=1}^d (c-\lambda_k) S^k\right)\\
&\quad + S^j\left(-\kappa \mu^i+ \sum_{k=1}^d A^i_k S^k +\mu^i\sum_{k=1}^d (c-\lambda_k) S^k\right)\\
&\quad -\gamma^{\mu}_{ij} S^i S^j\\
&=(\zeta+2c) S^i S^j+ S^i \left(\sum_{k=1}^d A^j_k S^k\right)+ S^j\left(\sum_{k=1}^d A^i_k S^k\right) -\gamma^{\mu}_{ij} S^i S^j.
\end{split}
\end{equation}
To guarantee that $\Sigma=\sum_{i=1}^d S^i$ we need to achieve $b^{\Sigma} = \sum_{i=1}^d b^S_i$ and $c^{\Sigma}= \sum_{i,j} c^S_{ij}$. Therefore it is necessary to impose $\kappa=0$ and 
\[
A^i_i=\lambda_i-c, \quad A^i_k=0, \quad i \neq k, 
\]
so that we finally obtain from \eqref{bS1}, \eqref{ciiS1}, \eqref{cijS1} the form of the characteristics stated in (ii).
Note that the matrix $\alpha$ in the characteristics of $S$ as specified in Proposition \ref{prop:statespace} is given by $\zeta\mathbf{1}\mathbf{1}^{\top}+\Lambda -\gamma^{\mu}$, where $\Lambda_{ij}=\lambda_i +\lambda_j$ and $\mathbf{1}$ denotes the vector with all entries equal to $1$. By Proposition~\ref{prop:statespace} we thus have the requirement that $\zeta\mathbf{1}\mathbf{1}^{\top}+\Lambda -\gamma^{\mu} +\Diag(\gamma^{\mu} s)\Diag(s)^{-1}$ is positive semidefinite for every $s \in \mathbb{R}^d_{++}$. As $\gamma^{\mu}_{ij} \geq 0$ for all $i\neq j$,  
\[
\Diag(\gamma^{\mu} s)\Diag(s)^{-1}=\Diag\left(\frac{1}{s_1}\sum_{j\neq 1} \gamma^{\mu}_{1j}s_j, \ldots,\frac{1}{s_d}\sum_{j\neq 1} \gamma^{\mu}_{dj}s_j\right)
\]
is positive definite. Each component $\frac{1}{s_i}\sum_{j\neq i} \gamma^{\mu}_{ij}s^j$ can be made arbitrarily small by choosing  $s_j$ for $j\neq i$ accordingly. 
This then yields that  $\zeta\mathbf{1}\mathbf{1}^{\top}+\Lambda -\gamma^{\mu} $ is positive semidefinite and $\zeta+ 2\lambda_i  \geq 0$ for all $i$, 
whence assertion (ii).

Conversely, assume (ii). Then the polynomial property of $S$ is clear. Moreover the characteristics of $(\mu, \Sigma)$ can be easily computed 
and are stated in (iii). From this  we see that $\mu$ is a Markovian It\^o-semimartingale with respect to $(\mathcal{F}_t)$ and thus also with respect to its natural filtration. Hence we can conclude the polynomial property of $\mu$ and thus (i).

The form of the characteristics as stated in (iii) follows from (i) and (ii) by simple computations. 
\end{proof}

\begin{proof}[Proof of Corollary \ref{cor:main}]
This is a simple consequence of Theorem \ref{th:main}. 
Indeed, on the one hand Condition (i) implies $\kappa=\phi=0$, whence we infer from (ii) in Theorem~\ref{th:main} that $\Sigma$ reduces to a Black-Scholes model. On the other hand Condition (ii) clearly also implies $\kappa=\phi=0$. From the form of the characteristics of $S$ as stated in item (iii) of Theorem \ref{th:main}, we therefore deduce that $S$ is polynomial with respect to its own filtration.
\end{proof}

\subsubsection{Proofs of Section \ref{sec:jumps}}

\begin{proof}[Proof of \ref{prop:jump}]
By independence of $\mu$ and $\Sigma$, the characteristics of $S^i=\mu^i\Sigma$ and the joint characteristics of $(\mu, S)$ with respect to the truncation function $\chi(\xi)=\xi$ read as follows
\begin{align*}
&b^{S}_i=\Sigma b^{\mu}_i + \mu^i b^{\Sigma}+ \int \xi_i^{\mu} \xi^{\Sigma} K(\mu, \Sigma, d\xi^{\mu}, d\xi^{\Sigma} ),\\
&c^{S}_{ij}=\Sigma^2c^{\mu}_{ij}+\mu^i \mu^j c^{\Sigma}, \quad c^{\mu, S}_{ij}=\Sigma c^{\mu}_{ij},\\
&K^{(\mu,S)}(G)=\\
&\, \,  \int 1_G\left( \xi^{\mu}_1,\ldots, \xi_d^{\mu}, \mu^1 \xi^{\Sigma}+\Sigma \xi^{\mu}_1+\xi^{\mu}_1\xi^{\Sigma}, \ldots,
\mu^d \xi^{\Sigma}+\Sigma \xi^{\mu}_d+\xi^{\mu}_d\xi^{\Sigma}\right)K(\mu, \Sigma, d\xi^{\mu}, d\xi^{\Sigma} ),
\end{align*}
where for $A \subseteq \mathbb{R}^{d+m}$, $K(\mu, \Sigma, A) =K(\mu, \{\xi^{\mu}: (\xi^{\mu},0) \in A\})+K(\Sigma, \{ \xi^{\Sigma}:  (0, \xi^{\Sigma}) \in A\})$ due to the independence of $\mu$ and $\Sigma$. This means in particular that $\mu$ and $\Sigma$ cannot jump together and that the jump term in the expression for $b^S$ vanishes. Hence, $b^{S}_i=\Sigma b^{\mu}_i + \mu^i b^{\Sigma}$ and is thus  a polynomial of degree $1$ in the components of $(\mu, S)$ again due to the independence of $\mu$ and $\Sigma$ and the form of the parameters stated in Proposition \ref{prop:statespace}.

Similarly we obtain due to \eqref{eq:jumpmeas}  that 
\begin{align*}
&c^{S}_{ij}+\int (\mu^i \xi^{\Sigma}+\Sigma \xi^{\mu}_i+\xi^{\mu}_i\xi^{\Sigma})(\mu^j \xi^{\Sigma}+\Sigma \xi^{\mu}_j+\xi^{\mu}_j\xi^{\Sigma}) K(\mu,\Sigma,d \xi^{\mu}, d\xi^{\Sigma})\\
&\quad=c^{S}_{ij}+\int (\mu^i)^2 (\xi^{\Sigma})^2 K(\Sigma, d\xi^{\Sigma})
+\int\Sigma^2 \xi^{\mu}_i\xi^{\mu}_j K(\mu, d\xi^{\mu})\\
&c^{\mu,S}_{ij}+\int (\mu^j \xi^{\Sigma}+\Sigma \xi^{\mu}_j+\xi^{\mu}_j\xi^{\Sigma}) \xi^{\mu}_i  K(\mu,\Sigma,d \xi^{\mu}, d\xi^{\Sigma})\\
&\quad=c^{\mu,S}_{ij}+\int \Sigma \xi^{\mu}_i\xi^{\mu}_j  K(\mu,d \xi^{\mu})
\end{align*}
are polynomials of degree $2$ in the components of $(\mu, S)$. Finally for $\mathbf{k}=(k_1, \ldots, k_{2d})$
\begin{align*}
&\int (\xi_1^{\mu})^{k_1}\cdots(\xi_d^{\mu})^{k_d}(\mu^1 \xi^{\Sigma}+\Sigma \xi^{\mu}_1+\xi^{\mu}_1\xi^{\Sigma})^{k_{d+1}} \cdots(\mu^d \xi^{\Sigma}+\Sigma \xi^{\mu}_d+\xi^{\mu}_d\xi^{\Sigma})^{k_{2d}}  K(\mu,\Sigma,d\xi^{\mu}, d\xi^{\Sigma}) \\
&\quad =\int (\xi_1^{\mu})^{k_1+k_{d+1}}\cdots(\xi_d^{\mu})^{k_d+k_{2d}} \Sigma^{\sum_{i=1}^d k_{d+i}}
K(\mu,d\xi^{\mu})\\
&\quad \quad +1_{\{(k_1,\ldots, k_d)=0\}} \int \mu_1^{k_{d+1}} \cdots \mu_d^{k_{2d}} (\xi^{\Sigma})^{\sum_{i=1}^d k_{d+i}} K(\Sigma, d\xi^{\Sigma})
\end{align*}
lies in $\mathcal{P}_{|\mathbf{k}|}(\Delta^d \times \mathbb{R}^d_+)$, i.e.~they are polynomials of degree $|\mathbf{k}|$  in the components of $(\mu, S)$. This shows the polynomial property of $(\mu,S)$, since the moment condition of Definition \ref{def:poly} holds as $\mu$ and $\Sigma$ are assumed to satisfy it.  
\end{proof}
\subsubsection{Proof of Proposition~\ref{prop:statespace}}\label{app:proof}

In order to prove Proposition~\ref{prop:statespace}, we start by the following auxiliary lemmas, which translate the assertions of \cite[Lemma E.1, E.3]{CLS:16} to the present state space $D=\Delta^d \times \mathbb{R}^m_+$ for $d \geq 2$.

\begin{lemma}\label{lem:representations}
Consider a polynomial $p \in \mathcal{P}_n$.
\begin{enumerate}
\item If $p$ vanishes on $\Delta^d \times \mathbb{R}^m_+ \cap \{x_i=0\}$ for some $i \in \{1, \ldots, d+m\}$, it can be written as
\begin{align}\label{eq:rep1}
p(x)=x_ip^i_{n-1}, \quad \text{ for some  }p^i_{n-1} \in \mathcal{P}_{n-1}.
\end{align}
\item If $p$ vanishes on $\Delta^d \times \mathbb{R}^m_+ \cap (\{x_i=0\} \cup \{x_j =0\})$ for some $i, j \in \{1, \ldots, d+m\}$, it can be written as
\begin{align}\label{eq:rep2}
p(x)=x_ix_jp^{ij}_{n-2}, \quad \text{ for some  } p^{ij}_{n-2} \in \mathcal{P}_{n-2}.
\end{align}
\end{enumerate}
\end{lemma}

\begin{proof}
Note that on the state space $\Delta^d \times \mathbb{R}^m_+$  every polynomial $p \in \mathcal{P}_n$ can be written as 
\[
p(x)=\sum_{|\mathbf{n}| =n} p_{\mathbf{n}} x^{\mathbf{n}},
\]
for real coefficients  $(p_{\mathbf{n}})_{|\mathbf{n}| =n}$. Indeed this is possible by multiplying with powers of $\sum_{i=1}^d x_i=1$. The assumption that $p$ vanishes on $\Delta^d \times \mathbb{R}^m_+ \cap \{x_i=0\}$ translates to
\[
0=p(x)=\sum_{|\mathbf{n}| =n, n_i=0} p_{\mathbf{n}} x^{\mathbf{n}},
\]
which implies that $p_{\mathbf{n}}=0$ for all multi-indices $\mathbf{n}$ such that $|\mathbf{n}| =n$ and $n_i=0$. We thus conclude that $p$ satisfies \eqref{eq:rep1}. Similarly we have for the second assertion (ii), that $p_{\mathbf{n}}=0$ for all multi-indices $\mathbf{n}$ such that $|\mathbf{n}| =n$ and $n_i=0$ or $n_j=0$. Therefore \eqref{eq:rep2} holds true.
\end{proof}

For the formulation of the subsequent lemma recall the notation  $I=\{1, \ldots, d\}$ and $J=\{d+1, \ldots, d+m\}$. Moreover, for a matrix $c \in \mathbb{S}^{d+m}$ we write $c_{II}$ for the matrix consisting of the first $d$ columns and rows and analogously $c_{JJ}$. Similarly, $x_I$ and $x_J$ stand for the vector $x$ consisting of the first $d$ and last $m$ elements respectively.

\begin{lemma} \label{lem:matrix} The following assertions are equivalent:
\begin{enumerate}
\item The matrix $c(x) \in \mathbb{S}^{d+m}$ satisfies $c(x)e_j=0$ on  $\Delta^d \times \mathbb{R}^m_+ \cap \{x^j=0\}$ for all $j \in I \cup J$, $c_{II}\mathbf{1}=0$ on $\Delta^d \times \mathbb{R}^m_+$  and $c_{ij} \in \mathcal{P}_2$ for all $i,j \in I \cup J$.
\item The matrix $c$ satisfies the condition stated in Proposition \ref{prop:statespace} (i).
\end{enumerate}
\end{lemma}

\begin{proof}
We start by proving (i) $\Rightarrow$ (ii) by applying similar arguments as in \cite[Lemma E.3]{CLS:16} and \cite[Proposition 6.4]{FL:14}. Since  $c(x)e_j=0$ on  $\Delta^d \times \mathbb{R}^m_+ \cap \{x^j=0\}$ implies $c_{ij}=0$ on $\Delta^d \times \mathbb{R}^m_+ \cap \{x_j=0\}$ and thus by symmetry $c_{ij}=0$ on $\Delta^d \times \mathbb{R}^m_+ \cap(\{x^i =0\} \cup \{x^j=0\})$. Lemma \ref{lem:representations} (ii) together with $c_{ij} \in \mathcal{P}_2$  thus yields $c_{ij}(x)=-\gamma_{ij}x_ix_j$ for all $i \neq j$ and some $\gamma_{ij} \in \mathbb{R}$. Moreover, as $c_{II}\mathbf{1}=0$ on $\Delta^d \times \mathbb{R}^m_+$, we also have that
\[
c_{ii}(x)= -\sum_{j \neq i, j \in I} c_{ij}(x)= \sum_{j \neq i, j \in I}\gamma_{ij}x_ix_j, \quad i \in I
\]
for all $x \in \Delta^d \times \mathbb{R}^m_+$. Since $c_{ii} \geq 0$ and 
$\gamma_{ij}$ can be written as $\gamma_{ij}=4c_{ii}(\frac{e_i+e_j}{2})$ it follows that $\gamma_{ij} \in \mathbb{R}_+$ for $i,j \in I$ and the form of $c_{II}$ is proved.

Let us now consider $c_{jj}$ for $j \in J$. By \ref{lem:representations} (i), we have 
$c_{jj}(x)=x_jp_1^j$ with some affine function $p_1^j$, which already yields the form
\[
c_{jj}(x)=\alpha_{jj}x_j^2+x_j(\phi_j+\theta_{(j)}^{\top} x_I+ \pi^{\top}_{(j)}x_J)
\]
with $\alpha_{jj} \in \mathbb{R}$, $\theta_{(j)} \in \mathbb{R}^d$, $\pi_{(j)} \in \mathbb{R}^m$ with $\pi_{(j)j}=0$. Positive semidefiniteness of $c(x)$ requires $c_{jj}(x) \geq 0$ for all $x$ on $\Delta^d \times \mathbb{R}_+^m$. This directly yields $\pi_{(j)} \in \mathbb{R}^m_+$. Furthermore by setting $x_k =0$ for $k \in J \setminus\{j\}$ and making $x_j$ sufficiently small, we see that $\phi_j +\theta^{\top}_{(j)}x_I\geq 0$ is required for all $x_I \in \Delta^d$, which forces $\phi_j \geq \max_{i\in I} \theta_{(j)i}^-$. Finally, let $\Theta \in \mathbb{R}^{d\times m}$ consist of the columns $\theta_{(j)}$ and  $\Pi \in \mathbb{R}^{m\times m}$ of the columns $\pi_{(j)}$. 
Moreover, let $\alpha \in \mathbb{S}^{m}$ with elements $\alpha_{ij}=-\gamma_{i+d,j+d}$,  $i,j \in \{1, \ldots, m\}$. We then have
\[
s^{-2}c_{JJ}(x_I, sx_J)=\Diag(x_J)\alpha\Diag(x_J)+\Diag(x_J)\Diag(s^{-1}(\phi+ \Theta^{\top} x_I) +\Pi^{\top} x_J).
\]
Letting $s \to \infty$, we see that $\alpha+ \Diag (\Pi^{\top} x_J)\Diag(x_J)^{-1} \in \mathbb{S}^m_+$ for all $x_J \in \mathbb{R}^m_{++}$, which leads the form of $c_{JJ}$.

It remains to show that $\gamma_{ij}=0$ for $i \in I$ and $j \in J$. Positive semidefiniteness of $c$ implies that 
\begin{align}\label{eq:CS}
c_{ii}(x)c_{jj}(x) \geq c_{ij}^2(x).
\end{align} 
Take now $x_i=s$ and $x_k=\frac{1-s}{d-1}$ for $I \ni k \neq i$. Then~\eqref{eq:CS} reads as
\[
c_{ii}(x)c_{jj}(x)=\frac{s(1-s)}{d-1} \left(\sum_{k\neq i, k \in I} \gamma_{ik}\right) (\gamma_{jj}x_j^2 +x_j(\phi_j+\theta_{(j)}^{\top} x_I+ \pi^{\top}_{(j)}x_J))
\geq \gamma_{ij}^2 s^2 x_j^2.
\]
For $s$ close to $1$, the left hand side can be made arbitrarily small so that the inequality is not satisfied if $\gamma_{ij} \neq 0$. This proves the first direction.

Concerning (ii) $\Rightarrow$ (i), the only condition which is not obvious is the positive semidefiniteness of $c_{II}$. This however follows exactly as in the proof of \cite[Proposition 6.6]{FL:14}.
\end{proof}

We are now ready to prove Proposition \ref{prop:statespace}.

\begin{proof}
We start by proving the first assertion on the necessary parameter conditions.
Being a polynomial process in the sense of Definition \ref{def:poly}, implies the well-posedness of the corresponding martingale problem in the sense of Remark \ref{rem:martingaleprob}. Hence, we can invoke \cite[Theorem 5.1]{FL:14}.
As in this paper we define the following set of polynomials $\mathscr{P}:=\{x_i \, | \, i=1,\ldots d+m\}$ and the following polynomial $q(x):=1-\sum_{j=1}^d x_j$ so that
\[
\Delta^d \times \mathbb{R}^m_+=\{x \in \mathbb{R}^{d+m}\, | \, p(x)\geq 0,  \, \forall p \in \mathscr{P}\} \cap M,
\]
where $M=\{x \in \mathbb{R}^{d+m}\, | \, q(x)=0 \}$. 
The conditions on $c(x)$ as stated in \cite[Theorem 5.1]{FL:14}, which are just consequence on the positive maximum principle, thus translate to
\begin{enumerate}
\item $c(x)e_j=0 \text{ on } \Delta^d \times \mathbb{R}^m_+ \cap \{x_j=0\} \text{ for all } j \in I \cup J$,
\item $c_{II}\mathbf{1}=0 \text{ on } \Delta^d \times \mathbb{R}^m_+$.
\end{enumerate}
This together with the polynomial property gives the conditions of Lemma \ref{lem:matrix} (i) and in turn the form of the instantaneous covariance matrix as stated in Proposition \ref{prop:statespace} (i).

Concerning part (ii), we have due the polynomial property $b(x)=\beta+Bx$ for some $\beta \in \mathbb{R}^{d+m}$ and $B\in \mathbb{R}^{(d+m) \times (d+m)}$.
Let us now consider the conditions of \cite[Theorem 5.1]{FL:14}, involving the drift part. To this end, denote by $\mathcal{G}$ the extended infinitesimal generator associated to the polynomial process as introduced in Remark \ref{rem:martingaleprob}
\[
\mathcal{G}f(x)=\sum_{i=1}^{d+m}D_if(x) b_i(x)+ \frac{1}{2} \sum_{i,j=1}^{d+m}D_{ij}f(x) c_{ij}(x).
\]
The condition $\mathcal{G}q=0$ on $\Delta^d \times \mathbb{R}_+^m$ of \cite[Theorem 5.1]{FL:14} therefore yields $\beta^{\top}_I\mathbf{1}+ x^{\top}B^{\top}(\mathbf{1}, \mathbf{0})^{\top}=0$ on $\Delta^d \times \mathbb{R}_+^m$. Thus it can be written as 
\[
\beta^{\top}_I\mathbf{1} + x^{\top}B^{\top}(\mathbf{1}, \mathbf{0})^{\top}=\kappa(1-\sum_{j=1}^d x_j)=0
\]
for some constant $\kappa$. This shows that $B_{IJ}=0$ and that $b_I(x)$ cannot depend on $x_J$. Moreover, we have $\beta^{\top}_I\mathbf{1}=\kappa$ and $B^{\top}_{II}\mathbf{1}=-\kappa\mathbf{1}=-(\beta^{\top}_I\mathbf{1}) \mathbf{1}$.
Furthermore, the second condition  of \cite[Theorem 5.1]{FL:14}, namely $\mathcal{G}p\geq 0$ on $\Delta^d \times \mathbb{R}^m_+ \cap \{p=0\}$ for all $p(x)=x_i, i \in I$ 
yields
\[
\beta_i+\sum_{k\neq i, k \in I} B_{ij}x_k \geq 0
\]
and we obtain, by inserting $e_j, I \ni j\neq i$, that $\beta_i+ B_{ij} \geq 0$ for $i,j \in I, i \neq j$.  Finally, for all $p(x)=x_j$ with $ j \in J$, the condition $\mathcal{G}p\geq 0$ on $\Delta^d \times \mathbb{R}^m_+ \cap \{p=0\}$ also has to be satisfied. We may thus set $x_J =0$ to see that $\beta_J+B_{JI}x_I$ has to lie in $ \mathbb{R}^m_+$ for all $x_I \in \Delta^d$. Hence $\beta_j \geq\max_{i \in I} B^-_{ji}$. Moreover, setting $x_j=0$ for some fixed $j \in J$ and letting $x_k \to \infty $, for $J \ni k \neq j$ forces $B_{jk} \geq 0$. This completes the proof of (ii).

\end{proof}

\begin{proof}[Proof of Proposition \ref{prop:wellposed}]
Note that the martingale problem for the components of $\Delta^d$ can be regarded separately since neither the drift nor the covariance depend on the factors in $\mathbb{R}^m_+$. The well-posedness in this case then follows e.g.~from \cite[Lemma 6.1]{CLS:16}. Combining this with the statement of \cite[Corollary 1.3]{BP:03} yields the assertion.
\end{proof}

\section{Proofs of Section \ref{sec:relarbitrage}}\label{app:relarbitrage}

Let us introduce the following notation needed in the subsequent proofs: 

\begin{notation}\label{notation}
For an admissible simplex parameter set $(\beta^{\mu}, B^{\mu}, \gamma^{\mu})$, let $c^{\mu}: \Delta^d \to \mathbb{S}^d_+$  and $b^{\mu}: \Delta^d \to \mathbb{R}^d$ be given by 
\[
b^{\mu}_i =\beta^{\mu}_i + \sum_{j=1}^d B^{\mu}_{ij} \mu^j,  \quad c_{ii}^{\mu}= \sum_{j \neq i} \gamma_{ij}^{\mu} \mu^i \mu^j, \quad \quad c^{\mu}_{ij} =\gamma_{ij}^{\mu}\mu^i \mu^j, \quad i \neq j.
\]
Note that we view here (and already previously) $c^{\mu}$ and $b^{\mu}$ as functions on $\Delta^d$ and $\mu^i$ stands for the $i^{\text{th}}$ component of a vector in $\Delta^d$. The latter will also be used in the subsequent proofs and it should be clear from the context if $\mu$ represents the process of market weights or rather an element in $\Delta^d$.  We write $b^{\mu}_t$ and $c^{\mu}_t$ when we insert $\mu_t$ in the above functions.
We denote by $\widetilde{c}^{\mu}$ the matrix $c^{\mu}$ with the $d^{\text{th}}$ row and column deleted and by $\widetilde{b}^{\mu}$ the vector $b^{\mu}$ with the $d^{\text{th}}$ entry deleted. 
\end{notation}

\begin{proof}[Proof of Theorem \ref{th:NANUPBR}]
We start to prove (ii) $\Rightarrow$ (i). Let us first show that there does not exist an equivalent measure $Q \sim P$ under which $\mu$ is a martingale. Assume by contradiction that such a martingale measure $Q$ exists. By Lemma \ref{lem:martingaleboundary}
 this in turn implies that every boundary segment $\{\mu^k=0\}$, $k \in \{1, \ldots,d\}$, is attained with positive $Q$ probability. But due to Proposition~\ref{prop:nonattainement}, Condition~\ref{eq:nonbd} is equivalent to the non-attainment of the boundary segment $\{\mu^i=0\}$ under $P$. Hence $P$ and $Q$ cannot be equivalent, whence by the Fundamental Theorem of Asset Pricing (see \cite{DS:94}), (NFLVR) is not satisfied. We now show that (NUPBR) is satisfied.
Indeed, by Lemma \ref{lem:formlambda} there exists some function $\widetilde{\lambda}$ such that 
\[
 \widetilde{b}^{\mu}= \widetilde{c}^{\mu}\widetilde{\lambda}(\mu)
\]
on $ E:=\{ \mu \in \Delta^d\,|\, \mu^j >0 \text{ for all } j \notin J\} $, where  $J$ denotes the set of indices $j$ for which $b^{\mu}_j=0 $ on $\{ \mu^j=0\}$ with $\widetilde{b}^{\mu}$, and $\widetilde{c}^{\mu}$ defined in \ref{notation}. Since $\mu$ takes values in $E$ as proved in Proposition~\ref{prop:nonattainement},  $\lambda(\mu_t)$ makes sense for all $t \in [0,T]$ and  implies that the so called weak structure condition (see \cite[Chapter 3]{HS:10}) is satisfied.
Moreover $\int_0^T \widetilde{\lambda}^{\top}(\mu_t)\widetilde{c}_t^{\mu} \widetilde{\lambda}(\mu_t) dt$ is $P$-a.s.~finite.  Indeed, as shown in the proof of Lemma \ref{lem:formlambda}, $\widetilde{\lambda}^{\top}(\mu_t)\widetilde{c}_t^{\mu} \widetilde{\lambda}(\mu_t) $ is given by 
\begin{align}\label{eq:equality1}
  \widetilde{\lambda}^{\top}(\mu_t)\widetilde{c}_t^{\mu} \widetilde{\lambda}(\mu_t) = (\widetilde{b}_t^{\mu})^{\top}(\widetilde{c}_t^{\mu})^+\widetilde{b}_t^{\mu}.
	\end{align}
	By Lemma \ref{eq:lemposdef},
we have $\widetilde{c}^{\mu}- \gamma^{\ast} \widetilde{a}^{\mu}$ is positive semidefinite 
where $\gamma^{\ast}=\min_{i\neq j} \gamma^{\mu}_{ij}$ and $\widetilde{a}^{\mu}$ is defined in Lemma \ref{lem:a}. Moreover, as also stated in Lemma \ref{eq:lemposdef}, the rank of both matrices $\widetilde{c}^{\mu}$ and  $\widetilde{a}^{\mu}$ is always the same.
By~\cite[Corollary 2]{W:79},  $1/\gamma^{\ast} (\widetilde{a}_t^{\mu})^{+}-(\widetilde{c}_t^{\mu})^{+}$ is therefore positive semidefinite and we can estimate
\begin{align}\label{eq:compare}
(\widetilde{b}_t^{\mu})^{\top}(\widetilde{c}_t^{\mu})^{+}\widetilde{b}_t^{\mu} \leq \frac{1}{\gamma^{\ast}} (\widetilde{b}_t^{\mu})^{\top}(\widetilde{a}_t^{\mu})^{+}\widetilde{b}_t^{\mu}. 
\end{align}
As on $\mathring{\Delta}^d$ the inverse matrix of $\widetilde{a}^{\mu}$ (and also $\widetilde{c}^{\mu}$) exists, we have  by Lemma \ref{lem:a} and the fact that $\sum_{=1}^{d-1} \widetilde{b}_j^{\mu}=-b^{\mu}_d$ (since the drift components $b^{\mu}$ have to sum up to $0$),  
\[
\frac{1}{\gamma^{\ast}} (\widetilde{b}_t^{\mu})^{\top}(\widetilde{a}_t^{\mu})^{-1}\widetilde{b}_t^{\mu}=\frac{1}{\gamma^{\ast}}\sum_{i=1}^d \frac{(b_{t,i}^{\mu})^2}{\mu_t^i} \quad \text{on } \mathring{\Delta}^d.
\]
By Lemma \ref{lem:representations} (i), we know that for all $j \in J$, $b^{\mu}_j= \kappa_j \mu^j$ for some constant $\kappa_j$. Hence 
\[
\frac{1}{\gamma^{\ast}} (\widetilde{b}_t^{\mu})^{\top}(\widetilde{a}_t^{\mu})^{-1}\widetilde{b}_t^{\mu}=\frac{1}{\gamma^{\ast}}\left(\sum_{j \in J} \kappa_j^2 \mu_t^j + \sum_{j \notin J} \frac{(b_{t,j}^{\mu})^2}{\mu_t^j}\right) \quad \text{on } \mathring{\Delta}^d.
\]
Extending this to $E$ yields the same equality with $(\widetilde{a}_t^{\mu})^{-1}$ replaced by $(\widetilde{a}_t^{\mu})^{+}$ and by \eqref{eq:compare} and \eqref{eq:equality1} we obtain 
\[
\widetilde{\lambda}^{\top}(\mu_t)\widetilde{c}_t^{\mu} \widetilde{\lambda}(\mu_t) \leq
\frac{1}{\gamma^{\ast}}\left(\sum_{j \in J} \kappa_j^2 \mu_t^j + \sum_{j \notin J} \frac{(b_{t,j}^{\mu})^2}{\mu_t^j}\right).
\]
As $\mu$ takes values in $E$, we can thus conclude that $\int_0^T \widetilde{\lambda}^{\top}(\mu_t)\widetilde{c}_t^{\mu} \widetilde{\lambda}(\mu_t) dt$ is $P$-a.s. finite as claimed. 
Hence, the so-called \emph{structure condition} (see e.g.~\cite[Chapter 3]{HS:10}) holds true and \cite[Theorem 3.4]{HS:10} thus implies that (NUPBR) holds for the process $\widetilde{\mu}$ and thus in turn for $\mu$ as well.
As (NFLVR) $\Leftrightarrow$ (NUPBR) + (NA) (see \cite{DS:94}) and since (NFLVR) does not hold, this thus means that relative arbitrages necessarily exist. 
Moreover, they are strong since the model is complete. Indeed, by Proposition \ref{prop:complete} and the existence of relative arbitrages, the $P$-a.s.~payoff $Y=1$ can be replicated with an initial capital strictly less than 1, given by  $E[Z_T] < 1$ where $Z$ is the unique strictly positive martingale deflator of Proposition \ref{prop:complete} (i).

Concerning the other direction. Assume (NUPBR) and that there exist strong relative arbitrage opportunities. Let us first prove that there must exist some $i \in \{1, \ldots,d\}$ such that $b^{\mu}_i > 0$ for some elements in $\{\mu^i=0\}$. Indeed, assume that for all $i$, $b^{\mu}_i = 0$ on $\{\mu^i=0\}$. By Lemma \ref{lem:representations} (i), this means that $b^{\mu}_i=\kappa_i \mu^i$ and since the index set $J=\{1, \ldots,d\}$, the set $E$ defined above is $\Delta^d$. 
Consider now the process 
\[
M_t:=-\int_0^t \widetilde{\lambda}^{\top}(\mu_s)\sqrt{\widetilde{c}^{\mu}_s} dW_s,
\] 
where $W$ is a $d-1$ dimensional Brownian motion and $\widetilde{\lambda}$ is defined in Lemma \ref{lem:formlambda}. Then the quadratic variation of $M$ is 
\[
\int_0^{\cdot}\widetilde{\lambda}^{\top}(\mu_s)\widetilde{c}_s^{\mu} \widetilde{\lambda}(\mu_s)ds.
\]
Analogously as above we obtain on $E=\Delta^d$ the estimate 
\[
\int_0^{\cdot}\widetilde{\lambda}^{\top}(\mu_s)\widetilde{c}_s^{\mu} \widetilde{\lambda}(\mu_s)ds\leq 
\frac{1}{\gamma^{\ast}} \int_0^{\cdot}\sum_{j =1}^d \kappa_j^2 \mu_s^j ds.
\]
Hence, Novikov's condition 
\[
E\left[e^{\frac{1}{2} \int_0^T 
\widetilde{\lambda}^{\top}(\mu_s)\widetilde{c}_s^{\mu} \widetilde{\lambda}(\mu_s)ds}\right]  < \infty
\]
is therefore satisfied and we conclude that $\mathcal{E}(M)$ is a martingale.
By Girsanov's theorem we thus obtain an equivalent measure $Q$ defined via $dQ/dP =\mathcal{E}(M_T)$ such that $b_i^{\mu}=0$  for all $i \in \{1, \ldots,d\}$ under $Q$. Hence $\mu$ is a martingale under $Q$ which contradicts the fact that there exist strong relative arbitrage opportunities. Thus there exist some $i \in \{1, \ldots,d\}$ such that $b^{\mu}_i > 0$ for some elements in $\{\mu^i=0\}$. Let now $i$ be such that  $b^{\mu}_i > 0$ for some elements in $\{\mu^i=0\}$. As (NUPBR) holds true it cannot happen with positive probability that $\mu^i_s=0$ for some $s \in [0,T]$ and that $\mu^i_t > 0$ for some $t >s$, because in this case an unbounded profit could be generated. This means that the drift must be strong enough to guarantee that the boundary segment  $\{\mu^i=0\}$ is not reached with positive probability.
By Proposition~\ref{prop:nonattainement} this is the case if and only if Condition~\ref{eq:nonbd} is satisfied. Hence (ii) holds.
\end{proof}

In order to prove Proposition \ref{prop:complete}, recall the martingale representation property or more generally the \emph{representation property relative to a semimartingale $\mu$ with continuous trajectories} (see ~\cite[Definition III.4.22]{JS:03}).

\begin{definition}
A local martingale $M$ has the \emph{representation property relative to $\mu$} if it has the form
\[
M_t=M_0 + \int_0^t h^{\top}_s d\mu_s^c, 
\]
where $\mu^c$ denotes the continuous martingale part and  $h\in L^2_{\mathrm{loc}}(\mu^c)$ 
(see~\cite[Definition III.4.3]{JS:03}).
\end{definition}

Polynomial diffusions on $\Delta^d$ have the representation property since the martingale problem is well-posed as stated subsequently.

\begin{proposition}\label{prop:representation}
Let $\mu$ be a polynomial diffusion process for the market weights on $\Delta^d$  being  described  by an admissible simplex parameter set $(\beta^{\mu}, B^{\mu}, \gamma^{\mu})$. Then all local martingales have the representation property with respect to $\mu$. 
\end{proposition}

\begin{proof}
By \cite[Lemma 6.1]{CLS:16} the martingale problem associated to a polynomial generator $\mathcal{G}$ defined via an admissible simplex parameter set $(\beta^{\mu}, B^{\mu}, \gamma^{\mu})$ in the sense of Remark \ref{rem:martingaleprob} is well posed. The assertion thus follows from \cite[Theorem III.4.29 (iv)$\Rightarrow$ (i)]{JS:03}.
\end{proof}

We are now ready to prove Proposition \ref{prop:complete}.

\begin{proof}[Proof of Proposition \ref{prop:complete}]
Concerning the first assertion, assumption (i) of Theorem \ref{th:NANUPBR} implies the existence of a strictly positive local martingale deflator $Z$ (see \cite{ST:14}). By strict positivity and Proposition \ref{prop:representation} there exists some process $\lambda$ such that $Z$ can be represented as $Z=\mathcal{E}(-\int_0^{\cdot} \lambda_s^{\top} \mu^c_s)$. As $Z\mu$ needs to be a local martingale, It\^o's formula implies that $\lambda$ satisfies
\begin{align}\label{eq:lambdaunique}
b^{\mu}_t= c^{\mu}_t \lambda_t, \quad P\text{-a.s.}
\end{align}
for almost all $t \in [0,T]$ (compare with Lemma \ref{lem:formlambda}). Although this equation does not necessarily have a unique solution $\lambda$, the quadratic variation of $Z$, given by $\int_0^{\cdot} \lambda_s^{\top} c^{\mu}_s \lambda_s ds$,  is uniquely determined since $\sqrt{c^{\mu}} \lambda$ is unique due to \eqref{eq:lambdaunique}.

Concerning part (ii), define
\[
M_t:= E[Y Z_T|\mathcal{F}_t],
\]
where $Z$ is the unique strictly positive local martingale deflator of part (i). By  Proposition \ref{prop:representation} there exists some strategy $h\in L^2_{\mathrm{loc}}(\mu^c)$  such that 
\[
M_t=M_0 + \int_0^t h^{\top}_s d\mu_s^c.
\]
Consider now the process $\frac{M_t}{Z_t}$. Then by It\^o's product rule and \eqref{eq:lambdaunique}, it is easily seen that
\[
\frac{M_t}{Z_t}= M_0+\int_0^t \frac{1}{Z_s} (h_s^{\top}+M_s\lambda_s^{\top}) d\mu_s.
\]
Setting $\psi_s=\frac{1}{Z_s} (h_s+M_s\lambda_s)$ and noticing that $M_0=E[Y Z_T]$ and $M_T/Z_T=Y$ yields the assertion.
\end{proof}

\begin{proof}[Proof of Lemma \ref{lem:martingaleboundary}]
Define for all $k \in \{1, \ldots, d\}$, $\tau_k=\inf\{t \geq 0 \,| \, \mu_t^k =0\}$.
We prove the assertion, which translates to
\begin{align}\label{eq:inductionhypothesis}
P_{\mu_0}[\tau_k \leq T] > 0,
\end{align}
for all $T >0$, $k \in \{1,\ldots,d\}$ and $\mu_0 \in \mathring{\Delta}^d$ by induction on the dimension $d$. 
For $d=2$, we have for $\mu_0 \in \mathring{\Delta}^2$
\begin{align*}
P_{\mu_0}[\tau_k \leq T] =(1-\mu_0^k) \left(1 - \int_0^1 p_T(\mu^1_0, dx)\right) >0, \quad k=1,2,
\end{align*}
where $p_T(\mu^1_0,dx)$ denotes the transition function of the one-dimensional Jacobi process without drift as for instance given in \cite[Equation 13.26]{KT:81}. This proves the assertion for $d=2$.
Let now $d \geq 3$. For the induction step from $d-1$ to $d$, we first show that
 the set $ A:=\{ \tau_k \leq \frac{T}{2} \textrm{ for some } k \}$ has positive probability.
Assume by contradiction that the boundary is not reached before $\frac{T}{2}$ and denote by $\mu^{(d)}=\min_i \mu^i$ and by $c^{\mu}_{(d)(d)}$ the instantaneous variance of  $\mu^{(d)}$. Then as in Remark \ref{rem:relarbitrage} we have
\begin{align*}
\frac{c^{\mu}_{(d)(d),t}}{\mu_t^{(d)}}
\geq\frac{d-1}{d} \min_{i \neq j} \gamma^{\mu}_{ij}, \quad t \in \left[0,\frac{T}{2}\right].
\end{align*}
By \cite{BF:08} (see also Remark \ref{rem:relarbitrage}), this condition implies the existence of relative arbitrages over the time horizon $[0,\frac{T}{2}]$, which however contradicts the fact that $\mu$ is a martingale and proves $P[A]>0$. Moreover it follows that there exists certainly some fixed index $k^{\ast} \in \{1, \ldots, d\}$ such that the set $ A_{k^{\ast}}:=\{ \tau_{k^{\ast}} = \min_k \tau_k \wedge \frac{T}{2} \}$ 
 has positive probability. Note that $A_{k^{\ast}}$ is simply the set of paths where the boundary segment $\{\mu^{{k^*}}=0\}$ is reached before or at the same time as the others and $T/2$. 
On the set $A_{k^{\ast}}$, $\tau_{k^{\ast}} < \tau_k$ for all $k \neq k^{\ast}$ a.s., since due to the non-degeneracy of $c^{\mu}$ in the interior of $\Delta^d$ (see Lemma \ref{lem:nondegenerate}) the probability to reach a $d-3$ dimensional manifold (corresponding to $\Delta^{d-2}$ and to a point for $d =3$) from the interior of $\Delta^d$ is  $0$. Denote by $\widetilde{\mu}$ the process where the $k^{\ast}$ component is removed. By the above argument it holds that on the set $A_{k^{\ast}}$, $\widetilde{\mu}_{\tau_{k^{\ast}}} \in \mathring{\Delta}^{d-1}$. We now apply the induction hypothesis conditional on $A_{k^{\ast}}$  which implies due to the strong Markov property of $\mu$ that for all $k \in \{1, \ldots, d-1\}$
\[
P_{\widetilde{\mu}_{\tau_{k^{\ast}}}}\left[ \widetilde{\tau}_k \leq \frac{T}{2}\,\Big|\, A_{k^{\ast}}\right]  > 0,
\]
where $\widetilde{\tau}_k=\inf\{t \geq 0 \,| \, \widetilde{\mu}_t^k =0\}$. 
Since $(A_j)_{j\in \{1, \ldots,d\}}$ together with $A^c$ is a partition of $\Omega$, we  have 
\begin{align*}
P_{\mu_0}[\tau_k \leq T] &= \sum_{j=1}^d P_{\mu_0}[\tau_k \leq T | A_{j}] P[A_j] + P_{\mu_0}[\tau_k \leq T | A^c] P[A^c] \\
&\geq P_{\widetilde{\mu}_{\tau_{k^{\ast}}}}\left[ \widetilde{\tau}_k \leq \frac{T}{2} \,\Big|\, A_{k^{\ast}}\right]  P[A_{{k^{\ast}}}] 1_{\{k \neq k^{\ast}\}}+ P[A_{{k^{\ast}}}] 1_{\{k = k^{\ast}\}} >0.
\end{align*}
\end{proof}

\begin{proof} [Proof of Proposition \ref{prop:nonattainement}]
To prove (ii) $\Rightarrow$ (i) we apply - similarly as in the proof of Lemma \ref{lem:strictpos} McKean's argument (see e.g.~\cite[Proposition 4.3]{MPS:11}) to the $\log(\mu^i)$. 
Then by It\^o's formula, we have for $t < \tau :=\inf\{ s \geq 0 \, |\, \mu^i_s =0\}$
\begin{align*}
\log(\mu_t^i)&=\log(\mu^i_0) + \int_0^t \frac{2 \beta_i^{\mu} + \sum_{j\neq i} (2B_{ij}^{\mu} -\gamma^{\mu}_{ij}) \mu_s^{j} + 2B^{\mu}_{ii}\mu_s^i}{2 \mu_s^i} \\
&\quad+\int_0^t \frac{\sum_{j=1}^d (\sqrt{c^{\mu}_s})_{ij}}{\mu^i_s} dW_s,
\end{align*}
where $W$ denotes some Brownian motion. Denote by $J \subseteq \{1, \ldots, i-1, i+1, \ldots,d\}$ the indices which satisfy $\arg \min_{i \neq j}(2B_{ij}^{\mu} -\gamma^{\mu}_{ij}) $. Then we can write 
\begin{align*}
\frac{2 \beta_i^{\mu} + \sum_{j\neq i} (2B_{ij}^{\mu} -\gamma^{\mu}_{ij}) \mu^{j} + 2B^{\mu}_{ii}\mu^i}{2 \mu_t^i}&= \frac{2 \beta_i^{\mu} + \min_{i \neq j} (2B_{ij}^{\mu} -\gamma^{\mu}_{ij})}{2\mu^i}\\
& +
\frac{(2B_{ii}^{\mu} -\min_{i \neq j} (2B_{ij}^{\mu} -\gamma^{\mu}_{ij}))}{2}\\
& + \sum_{j \in J^c \setminus \{i\}} (2B_{ij}^{\mu} -\gamma^{\mu}_{ij}- \min_{i \neq j} (2B_{ij}^{\mu} -\gamma^{\mu}_{ij})) \frac{\mu^j}{2\mu^i}.
\end{align*}
Since the first and the last term are nonnegative and the second is constant, we deduce that for every $\mathbf{T} >0$ 
\[
\inf_{t \in [0, \tau \wedge \mathbf{T})} \int_0^t \frac{2 \beta_i^{\mu} + \sum_{j\neq i} (2B_{ij}^{\mu} -\gamma^{\mu}_{ij}) \mu_s^{j} + 2B^{\mu}_{ii}\mu_s^i}{2 \mu_s^i} > -\infty, \quad P\text{-a.s.}
\]
Mc Kean's argument as of \cite[Proposition 4.3]{MPS:11} therefore yields that $\tau =\infty$ and thus implies the assertion.

For the converse direction we apply \cite[Theorem 5.7 (iii)]{FL:14} and assume for a contradiction that (ii) does not hold, i.e. 
\[
 2 \beta_i^{\mu} + \min_{i \neq j}(2B_{ij}^{\mu} -\gamma^{\mu}_{ij}) <0.
\]
In the terminology of \cite[Theorem 5.7 (iii)]{FL:14}, we have $p(\mu)=\mu^i$,
such that 
\[
c^{\mu} \nabla p =\left(\begin{array}{c} -\gamma^{\mu}_{1i} \mu^1 \mu^i\\
\vdots\\
\sum_{j\neq i} \gamma^{\mu}_{ij} \mu^{j} \mu^{i}\\
\vdots\\
-\gamma^{\mu}_{di} \mu^d \mu^i
\end{array}\right)= \mu^i h(\mu),
\]
with 
$h=(-\gamma^{\mu}_{1i} \mu^1, \cdots,\sum_{j\neq i} \gamma^{\mu}_{ij} \mu^{j}, \cdots,-\gamma^{\mu}_{di} \mu^d )^{\top}$.
Let $j^*$ be  the $\arg \min_{j \neq i} (2B_{ij}^{\mu} -\gamma^{\mu}_{ij})$ and let $\bar{\mu}=e_{j^*}$.  
Then $\mathcal{G} p(\bar{\mu})= \beta^{\mu}_i+ B_{ij^*} \geq 0$ by Definition \ref{def:parameters} and
\begin{align*}
2 \mathcal{G} p(\bar{\mu})- h^{\top} \nabla p(\bar{\mu}) = 2 \beta_i^{\mu} +\min_{i \neq j} 2(B_{ij}^{\mu} -\gamma^{\mu}_{ij})<0.
\end{align*}
By \cite[Theorem 5.7 (iii)]{FL:14}, it thus follows that for any time horizon $T$ we can find some $\mu_0 \in \mathring{\Delta}^d$ close enough to $\bar{\mu}$ such that $\{\mu^i =0\}$ is hit with positive probability.
This contradicts (i) and proves the assertion.
\end{proof}

\begin{lemma}\label{lem:nondegenerate}
Let $\gamma^{\mu}$ be such that $\gamma^{\mu}_{ij}>0$ for all $i \neq j$. Then the matrix $\widetilde{c}^{\mu} \in \mathbb{S}^{d-1}_{++}$ and $c^{\mu}$ has rank $d-1$ for all $\mu \in \mathring{\Delta}^d$.
\end{lemma}
\begin{proof}
Note that the matrix $\widetilde{c}^{\mu}$  is strictly diagonally dominant, i.e. $|\widetilde{c}^{\mu}_{ii}| > \sum_{j \neq i} |\widetilde{c}^{\mu}_{ij}|$, if $\mu \in \mathring{\Delta}^d$. Indeed, we have
\[
|\widetilde{c}^{\mu}_{ii}| =\sum_{j\neq i} \gamma^{\mu}_{ij} \mu^i \mu^j >\sum_{j\neq \{i, d\}  } \gamma_{ij} \mu^i \mu^j=
\sum_{j \neq i} |\widetilde{c}^{\mu}_{ij}|.
\]
By \cite[Theorem 6.1.10]{HJ:85}, $\widetilde{c}^{\mu}$ is thus strictly positive definite which implies that $c^{\mu}$ has rank $d-1$.
\end{proof}

\begin{remark}
By similar arguments, it actually suffices to assume that there exists some index $j$ such that $\gamma^{\mu}_{ij} >0$ for all $i$ in order to obtain the same conclusion.
\end{remark}

\begin{lemma}\label{lem:formlambda}
Let $\gamma^{\mu}$ be such that $\gamma^{\mu}_{ij}>0$ for all $i \neq j$ and $b^{\mu}$ a drift being described by admissible parameters $(\beta^{\mu}, B^{\mu})$. Let $J=\{j_1, \ldots, j_k\}$, $0 \leq k \leq d$,  denote the set of indices for which $b^{\mu}_{j_i}=0$ on $\{\mu^{j_i} =0\}$ and set $E:=\{ \mu \in \Delta^d\,|\, \mu^j >0 \text{ for all } j \notin J\}$.
Then there exists some function $\widetilde{\lambda}: E \to \mathbb{R}^{d-1}$ such that 
\begin{align} \label{eq:bmulambda}
 \widetilde{b}^{\mu}= \widetilde{c}^{\mu}\widetilde{\lambda}(\mu)
\end{align}
on $ E $.
\end{lemma}

\begin{proof}
Define $\widetilde{\lambda}$ as follows
\begin{align}\label{eq:lambdatilde}
\widetilde{\lambda}(\mu)=(\widetilde{c}^{\mu})^+\widetilde{b}^{\mu},
\end{align}
where $(\widetilde{c}^{\mu})^+$ denotes the Moore-Penrose pseudo-inverse. Let $O^{\mu} D^{\mu} (O^{\mu})^{\top}$ be the spectral decomposition of $\widetilde{c}^{\mu}$ with orthogonal matrices $O^{\mu}$ and a diagonal matrix $D^{\mu}$. We write here the superscript $\mu$ to indicate the dependence on $\mu$. Then  $(\widetilde{c}^{\mu})^+$ is given by
\[
(c^{\mu})^+=   O^{\mu}(D^{\mu})^+ (O^{\mu})^{\top},
\]
where $(D^{\mu})_{ii}^+= \frac{1}{ D^{\mu}_{ii}}$ if $D^{\mu}_{ii} \neq 0$ and $0$ otherwise.
By Lemma \ref{lem:nondegenerate}, $\widetilde{c}^{\mu}$ is invertible if $\mu \in \mathring{\Delta}^{d}$, and $(\widetilde{c}^{\mu})^+$ thus coincides with $(\widetilde{c}^{\mu})^{-1}$. This already proves the claim on $\mathring{\Delta}^{d}$.

It thus remains to show that $\widetilde{\lambda}$ given by \eqref{eq:lambdatilde} also satisfies \eqref{eq:bmulambda} whenever $\mu^{j_i}=0$ for some $j_i \in J$. In this case  $b^{\mu}_{j_i}=0$ by assumption and $ c^{\mu}_{j_i,k}=0$ for all $k \in \{1, \ldots, d\}$ due to the form of $c^{\mu}$. Note in particular that 
for $\mu^{d}=0$  the claim already follows from the considerations above.
Let us choose the order of the eigenvectors in $O^{\mu}$ such that $D^{\mu}_{j_i j_i}=0$ whenever $\mu^{j_i}=0$. Inserting now \eqref{eq:lambdatilde} into \eqref{eq:bmulambda} yields
\[
 \widetilde{b}^{\mu}=\widetilde{c}^{\mu}\widetilde{\lambda}(\mu)= \widetilde{c}^{\mu}(\widetilde{c}^{\mu})^+\widetilde{b}^{\mu}=O^{\mu} I^{\mu} (O^{\mu})^{\top} \widetilde{b}^{\mu},
\]
where $I^{\mu}$ is a diagonal matrix with all entries equal to $1$ except of the indices $j_i$ where $\mu^{j_i}=0$. Thus \eqref{eq:bmulambda} is equivalent to
\[
(O^{\mu})^{\top}\widetilde{b}^{\mu}=I^{\mu} (O^{\mu})^{\top} \widetilde{b}^{\mu}
\]
and this equality holds true since $((O^{\mu})^{\top} \widetilde{b}^{\mu})_{j_i}=0$. Indeed,  
\[
0=\widetilde{c}^{\mu}_{j_i j_i} = \sum_{k=1}^d (O^{\mu}_{j_i k})^2 D^{\mu}_{kk}
\]
implies that $O_{j_i k}=0$ for whenever $D^{\mu}_{kk} \neq 0$. As $b_k = 0$ whenever $D^{\mu}_{kk}= 0$, this yields the claim $((O^{\mu})^{\top}\widetilde{b}^{\mu})_{j_i}=0$ and in turn the  assertion of the lemma.
\end{proof}

\begin{lemma}\label{lem:a}
Let $\gamma^{\mu}$ be such that $\gamma^{\mu}_{ij}=1$ for all $i \neq j$ and denote $\widetilde{c}^{\mu}$ for this particular form of $\gamma$ by $ \widetilde{a}^{\mu}$. Then, 
 on $\mathring{\Delta}^d$, the entries of the inverse of $\widetilde{a}^{\mu}$ are given by
\begin{align}\label{eq:inv}
(\widetilde{a}^{\mu})_{kk}^{-1}=\frac{1}{\mu^k}+\frac{1}{ \mu^d}, \quad  (\widetilde{a}^{\mu})_{kl}^{-1}=\frac{1}{\mu^d}, \quad k \neq l.
\end{align}
\end{lemma}

\begin{proof}
The form of the inverse is easily verified by a simple computation.
\end{proof}

\begin{remark}
Note that the above form of $\gamma^{\mu}$ yields the diffusion matrix of volatility stabilized models with entries specified in Proposition \ref{prop:Jacobi}. Moreover, observe that in this case $\widetilde{a}^{\mu}$  corresponds to $c^{\mu}$ of dimension $d-1$.
\end{remark}

\begin{lemma}\label{eq:lemposdef}
Let $\gamma^{\mu}$ be such that $\gamma^{\mu}_{ij}>0$ for all $i \neq j$. Set $\gamma^{\ast}=\min_{i \neq j} \gamma^{\mu}_{ij}$ and let $\widetilde{a}^{\mu}$ be given as in Lemma \ref{lem:a}. Then $
\widetilde{c}^{\mu}- \gamma^{\ast} \widetilde{a}^{\mu}$
is positive semidefinite. Moreover, $\rk(\widetilde{c}^{\mu})=\rk(\widetilde{a}^{\mu})$.
\end{lemma}

\begin{proof}
Note that the matrix $
\widetilde{c}^{\mu}- \gamma^{\ast} \widetilde{a}^{\mu}$ corresponds to a diffusion matrix (with last row and column deleted) generated by some matrix $\widehat{\gamma}^{\mu}$ given by
\[
\widehat{\gamma}^{\mu}_{ij}=\gamma_{ij} -\gamma^{\ast}, \quad i \neq j.
\]
Since all entries of $\widehat{\gamma}^{\mu}$ are nonnegative, positive semidefiniteness of $c^{\mu}- \gamma^{\ast} a^{\mu}$ and thus in turn of $\widetilde{c}^{\mu}- \gamma^{\ast} \widetilde{a}^{\mu}$ follows. 
For the last statement, note that the rank of both matrices is equal to the dimension of the corresponding boundary segment of the simplex, i.e. $d-1$ if $\mu \in \mathring{\Delta}^d$, etc.
\end{proof}

\begin{proof}[Proof of Proposition \ref{prop:arbitrageexplicit}]
We follow the idea outlined in Remark \ref{rem:optstrat}. Indeed, consider a sequence of polynomials $(p_{n})_{n \in \mathbb{N}}$ on $\Delta^d$ approximating  $\mu \mapsto 1_E(\mu)$ pointwise such that $p_n(\Delta^d \setminus E)=0$ for all $n$. Let $Z$ be the strictly positive local martingale deflator as of Proposition \ref{prop:complete} and define 
\[
q:=E[Z_T] <1,
\]
which corresponds to the superheding price $U_T$ of $1$, given in Definition \ref{def:superhedge} and which is strictly smaller than $1$ since relative arbitrages exist. 
Let $\delta >0$ such that $q+\delta <1$. By dominated convergences (choosing the sequence of polynomials bounded) there exists some $n$ such that for all $N \geq n$
\[
\left|E\left[p_N(\mu_T) Z_T\right] -q\right|\leq \delta.
\]
Let $\varepsilon >0$. Since $p_n(\mu) \to  1_{E}(\mu)$  as $n \to \infty$ and as $P$ has no mass outside the set $E$ there exists some $\widetilde{n}$ such that for $N \geq \widetilde{n}$ 
\[
P[p_{N}(\mu_T) > q+\delta] \geq 1-\varepsilon.
\]
Take now $n^{\varepsilon}=\max(n, \widetilde{n})$, which implies 
\[
P\left[\frac{p_{n^{\varepsilon}}(\mu_T)}{E\left[p_{n^{\varepsilon}}(\mu_T)Z_T\right] }>1\right] \geq 1-\varepsilon.
\]
Our goal is now to find a strategy $\vartheta^{\varepsilon}$ as given in the statement of the proposition such that 
\[
Y_T^{\vartheta^{\varepsilon}}= \frac{p_{n^{\varepsilon}}(\mu_T)}{E\left[p_{n^{\varepsilon}}(\mu_T)Z_T\right]}.
\]
To this end, let $Q$ denote the  F\"ollmer measure, as already introduced in Remark \ref{rem:optstrat}, satisfying $P \ll Q$ and under which $\mu^{\tau}$ is a martingale,  where $\tau=\inf\{ t >0 \, | \, \mu_t \notin E\} = \inf\{ t >0 \, | \, \frac{1}{Z_t} =0\}$. 
Consider furthermore a polynomial martingale $\hat{\mu}$ with the same covariance structure as $\mu$ and denote its law on path space by $\hat{Q}$. Then, by uniqueness of the corresponding martingale problem, the laws of $\mu^{\tau}$ and $\hat{\mu}^{\tau}$ coincide.
Define now 
\[
p_{n^{\varepsilon}}(t ,\hat{\mu}_t):=E_{\hat{Q}}[p_{n^{\varepsilon}}(\hat{\mu}_T)|\mathcal{F}_t].
\]
By the polynomial property of $\hat{\mu}$, $\mu \mapsto p_{n^{\varepsilon}}(t ,\mu)$ is a polynomial and in particular sufficiently regular in both variables to apply It\^o's formula which yields
\[
p_{n^{\varepsilon}}(\hat{\mu}_T)=p_{n^{\varepsilon}}(T,\hat{\mu}_T)=p_{n^{\varepsilon}}(0,\hat{\mu}_0)+ \int_0^T \sum_{i=1}^d D_i p_{n^{\varepsilon}}(t,\hat{\mu}_t) d\hat{\mu}_t, \quad \hat{Q}\text{-a.s.}
\]
since  $(p_{n^{\varepsilon}}(t ,\hat{\mu}_t))_{t \in [0,T]}$ and $\hat{\mu}$ are $\hat{Q}$-martingales. Clearly, by stopping at time $\tau$ we have
\[
p_{n^{\varepsilon}}(\hat{\mu}_{T \wedge \tau})=p_{n^{\varepsilon}}(0,\hat{\mu}_0)+ \int_0^{T\wedge \tau} \sum_{i=1}^d D_i p_{n^{\varepsilon}}(t,\hat{\mu}_t) d\hat{\mu}_t, \quad \hat{Q}\text{-a.s.}
\]
and the same holds true by removing the ``hats'' since the laws of $\mu$ and $\hat{\mu}$ coincide on the stochastic interval $[0, \tau \wedge T]$.  Since $P\ll Q$ and $P[ \tau >T]=1$, this implies
\[
p_{n^{\varepsilon}}(\mu_T)=p_{n^{\varepsilon}}(0,\mu_0)+ \int_0^{T} \sum_{i=1}^d D_i p_{n^{\varepsilon}}(t,\mu_t) d\mu_t, \quad P\text{-a.s.}
\]
Define now $\vartheta^{i,\varepsilon}$ by 
\begin{align}\label{eq:vartheta}
\vartheta^{i,\varepsilon}=\frac{D_i p_{n^{\varepsilon}}(t,\mu_t)}{E\left[p_{n^{\varepsilon}}(\mu_T)Z_T\right]}
\end{align}
and note that 
\[ 
E\left[p_{n^{\varepsilon}}(\mu_T)Z_T\right]=E_Q\left[p_{n^{\varepsilon}}(\mu_T)1_{\{\frac{1}{Z_T} >0\}}\right]=E_{\hat{Q}}\left[p_{n^{\varepsilon}}(\hat{\mu}_T)\right]=p_{n^{\varepsilon}}(0,\mu_0),
\]
where the first equality follows from the generalized Bayes rule (see e.g.~\cite[Theorem 5.1]{R:13}) and the second from the fact that the laws of $\mu^{\tau}$ and $\hat{\mu}^{\tau}$ coincide and that $p_{n^{\varepsilon}}(\Delta^d \setminus E)=0$.
This thus yields
\[
Y_T^{\vartheta^{\varepsilon}}=\frac{p_{n^{\varepsilon}}(\mu_T)}{E\left[p_{n^{\varepsilon}}(\mu_T)Z_T\right]}=1+\int_0^T \sum_{i=1}^d \vartheta^{i,\varepsilon} d\mu_t
\]
and therefore
the first assertion, since $\mu \mapsto p_{n^{\varepsilon}}(t,\mu)$ is a time dependent polynomial.
Concerning the second one, $Y_T^{\vartheta^{\varepsilon}}$ 
clearly converges $P$-a.s.~to the optimal arbitrage given by $\frac{1}{U_T}=\frac{1}{E[Z_T]}$, since $p_n( \mu_T) \to 1$ $P$-a.s., as $P$ has no mass outside $E$.
\end{proof}

%\bibliographystyle{abbrv}

%\bibliography{referencesSPT150404,170504bibl}

\end{document}